\documentclass{ws-ijfcs}

\usepackage{amsmath,amssymb,float}

\def\Enn{{\mathbb N}}
\def\Que{{\mathbb Q}}
\def\Zee{{\mathbb Z}}

\def\fmod#1 #2{#1\ ({\rm mod}\ #2)}
\def\Enni{{\Enn_{\infty}}}
\def\had{\odot}

\usepackage[dvips]{epsfig}
\usepackage{epsf}

\begin{document}

\markboth{\'Emilie Charlier, Narad Rampersad, and Jeffrey Shallit}{Enumeration and Decidable Properties of Automatic Sequences}

\title{ENUMERATION AND DECIDABLE PROPERTIES OF AUTOMATIC SEQUENCES}

\author{\'EMILIE CHARLIER}
\address{School of Computer Science,
University of Waterloo,
Waterloo, ON  N2L 3G1,
Canada\\
\email{\tt emilie.charlier@ulb.ac.be}}

\author{NARAD RAMPERSAD}
\address{Department of Math/Stats,
University of Winnipeg,
Winnipeg, MB, R3B 2E9,
Canada\\
\email{\tt narad.rampersad@gmail.com}}

\author{JEFFREY SHALLIT}
\address{School of Computer Science,
University of Waterloo,
Waterloo, ON  N2L 3G1,
Canada\\
\email{\tt shallit@cs.uwaterloo.ca}}

\maketitle

\newtheorem{openproblem}[theorem]{Open Problem}

\begin{abstract}
We show that various aspects of $k$-automatic sequences --- such as having an 
unbordered factor of length $n$ --- are both
decidable and effectively enumerable.  
As a consequence it follows that
many related sequences are either $k$-automatic or $k$-regular.
These include many sequences previously studied in the literature,
such as the recurrence function, the appearance function, and the
repetitivity index.   
We also give some new characterizations of
the class of $k$-regular sequences.  Many results extend to
other sequences defined in terms of Pisot numeration
systems.
\end{abstract}

\section{Introduction}

Let ${\bf x} = (a(n))_{n \geq 0}$ be an infinite sequence over a finite
alphabet $\Delta$.  We write ${\bf x}[i] = a(i)$, and we let
${\bf x}[i..i+n-1]$ denote the factor of length $n$ beginning at position $i$.
    
An infinite
sequence $\bf x$ is
said to be {\it $k$-automatic} if it is computable by a finite automaton
taking as input the base-$k$ representation of $n$, and
having $a(n)$ as the output associated with the last state encountered
\cite{Allouche&Shallit:2003b}.

For example, in Figure~\ref{fig1}, we see an automaton generating the
Thue-Morse sequence ${\bf t} = t_0 t_1 t_2 \cdots = {\tt 011010011001} \cdots$.
The input is $n$, expressed in base $2$, and the output is the number contained in
the state last reached.

\begin{figure}[H]
\leavevmode
\def\epsfsize#1#2{1.0#1}
\centerline{\epsfbox{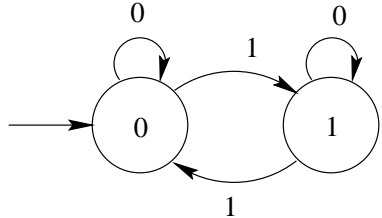}}
\protect\label{fig1}
\caption{A finite automaton generating a sequence}
\end{figure}

     Honkala \cite{Honkala:1986} showed that, given an automaton,
it is decidable if the sequence it generates is ultimately periodic.
Later, Leroux \cite{Leroux:2005} gave a polynomial-time algorithm for the problem.

Recently, Allouche, Rampersad, and Shallit
\cite{Allouche&Rampersad&Shallit:2009} found a different proof of
Honkala's result using a more general technique.  They showed that their
technique suffices to show that the following
properties (and many more) are decidable for $k$-automatic sequences $\bf x$:

\begin{itemize}
\item[(a)] Given a rational number $r > 1$, whether $\bf x$ is $r$-power-free;
\item[(b)] Given a rational number $r > 1$, whether $\bf x$ contains infinitely
many occurrences of $r$-powers;
\item[(c)] Given a rational number $r > 1$, whether $\bf x$ contains infinitely
many distinct $r$-powers;
\item[(d)] Given a length $l$, whether $\bf x$ avoids palindromes of length
$\geq l$.
\end{itemize}

Related
results have recently been given by
Halava, Harju, K\"arki, and Rigo 
\cite{Halava&Harju&Karki&Rigo:2010}.

In this paper we show that many additional 
properties of automatic
sequences are decidable using the same general technique.
More significantly,
we also show that related enumeration questions on automatic
sequences (such as counting
the number of distinct factors of length $n$) can be solved using
a similar technique, in an entirely effective manner.
As a consequence, we recover or improve
results due to Moss\'e \cite{Mosse:1996a};
Allouche, Baake, Cassaigne, and Damanik
\cite{Allouche&Baake&Cassaigne&Damanik:2003};
Currie and Saari \cite{Currie&Saari:2009};
Garel \cite{Garel:1997}; Fagnot \cite{Fagnot:1997a};
and Brown, Rampersad, Shallit, and Vasiga
\cite{Brown&Rampersad&Shallit&Vasiga:2006}.

Our main results about decidability are given in Section~\ref{logicsec}, and
our main results about enumeration are given in Section~\ref{applications}.

Throughout this paper, $k$ denotes a fixed integer $\geq 2$, the symbol
$\Enn$ denotes the non-negative integers
$\lbrace 0,1,2, \ldots \rbrace$, and the symbol
$\Enni$ denotes the ``extended'' non-negative integers
$\Enn \ \cup \ \lbrace \infty \rbrace$.

\section{Connections with logic and new decidability results}
\label{logicsec}

After the publication of \cite{Allouche&Rampersad&Shallit:2009}, the
third author noticed that the technique used there was, at its core,
very similar to previous techniques developed by B\"uchi, Bruy\`ere,
Michaux, Villemaire, and others, involving formal logic; see, e.g.,
\cite{Bruyere&Hansel&Michaux&Villemaire:1994}.  This was later
independently observed by the first author, as well as by V\'eronique
Bruy\`ere.  As it turns out, the properties (a)--(d) above are
decidable because they are expressible as predicates in the first-order
structure $\langle \Enn, +, V_k \rangle$, where $V_k (n)$ is the
largest power of $k$ dividing $n$.

We briefly recall the technique discussed in
\cite{Allouche&Rampersad&Shallit:2009} in the context of a particular
example.  Suppose we want to decide if an automatic sequence $\bf x$
is squarefree (contains no nonempty square factor).
Given an automaton $M$ generating
a $k$-automatic sequence $\bf x$, we create, via a series of transformations,
a new automaton $M'$ that accepts the base-$k$ representations of integers
corresponding to the squares in $\bf x$.  For example, $M'$ could accept
those integers corresponding to the starting position of each square, or those
integers corresponding to the lengths of the squares.  The operations
we can use in constructing $M'$ include digit-by-digit addition or 
subtraction (with carry, if necessary), comparison, and lookup of the
corresponding term in $\bf x$ (which comes from simulation of $M$).
Nondeterminism can be used to implement ``$\exists$'', and ``$\forall$''
can be implemented by nondeterminism combined with suitable negations.

Ultimately, then,
deciding if $\bf x$ is squarefree corresponds to verifying that
$L(M') = \emptyset$ for the $M'$ we construct.
Deciding whether $\bf x$ contains only finitely
many square occurrences corresponds to verifying that $L(M')$ is
finite.  Both can easily be done by the standard methods for automata, using
depth-first or breadth-first search on the underlying state diagram of the
automaton.

In this paper, we always assume that numbers are encoded in base $k$ using
the digits in $\Sigma_k = \lbrace 0, 1, \ldots, k-1 \rbrace$,
The {\it canonical encoding}
of $n$ is the one with no leading zeroes 
and is denoted $(n)_k$.    Similarly, if $w = a_1 \cdots a_n
\in \Sigma_k^*$, then
by $[w]_k$ we mean $\sum_{1 \leq i \leq n} a_i k^{n-i}$, the integer
that $w$ represents.
Often we will deal with reversed representations, where the least
significant digit appears first.
For example, in the reversed representation,
$13$ is represented in base $2$ by the word $1011$.  

Sometimes
we will need to encode pairs, triples, or $r$-tuples of
integers.  We handle these
by first padding the reversed representation of
the smaller integer with trailing zeroes,
and then coding the $r$-tuple as a word over $\Sigma_k^r$.
For example, the pair $(20,13)$ could be represented in base-$2$ as
$$ [0,1][0,0][1,1][0,1][1,0] ,$$
where the first components spell out $00101$ and the second components
spell out $10110$. Of course,
there are other possible representations, such as
$$ [0,1][0,0][1,1][0,1][1,0][0,0],$$
which correspond to non-canonical
representations having trailing zeroes.  In general, we permit these.

Thus, the main idea of \cite{Allouche&Rampersad&Shallit:2009} can be
restated as follows:

\begin{theorem}
If we can express a property of a $k$-automatic sequence $\bf x$ using quantifiers,
logical operations, integer variables, the operations
of addition, subtraction, indexing into $\bf x$, and comparison of integers or elements
of $\bf x$, then this property is decidable.
\label{logic}
\end{theorem}

We illustrate the idea with the following new result.  A word $w$ is
{\it bordered} if it begins and ends with the same word $x$ with
$0 < |x| \leq |w|/2$.  (An example in English is {\tt ingoing}, which begins
and ends with {\tt ing}.)  
Otherwise it is \textit{unbordered}.  

\begin{theorem}
Let ${\bf x} = a(0) a(1) a(2) \cdots$ be a $k$-automatic sequence.
Then the associated infinite sequence
${\bf b} = b(0) b(1) b(2) \cdots$ defined by
$$
b(n) = \begin{cases}
	1, & \text{if $\bf x$ has an unbordered factor
	of length $n$;} \\
	0, & \text{otherwise;}
	\end{cases}
$$
is $k$-automatic.
\end{theorem}

\begin{proof}
The sequence $\bf x$ has an unbordered factor of length $n$

\smallskip

iff

\smallskip

\noindent $\exists j \geq 0$ such that the factor of length $n$
beginning at position $j$ of $\bf x$ is unbordered

\smallskip

iff

\smallskip

\noindent
there exists an integer $j \geq 0$ such that for all possible lengths $l$
with $1 \leq l \leq n/2$, there is an integer $i$ with $0 \leq i < l$
such that the $i$'th
letter in the supposed border of length $l$ beginning and ending the 
factor of length $n$ beginning at position $j$ of $\bf x$
actually differs in the $i$'th position

\smallskip

iff

\smallskip

\noindent there exists an integer $j \geq 0$ such that for all integers
$l$ with $1 \leq l \leq n/2$
there exists an integer $i$ with $0 \leq i < l$ such
that $a(j+i) \not= a(j+n-l+i)$.

To carry out this test, we first create an NFA that given
the encoding of $(j,l,n)$ guesses the base-$k$
representation of $i$, digit-by-digit, checks that $i < l$,
computes $j+i$ and $j+n-l+i$ on the fly, and 
checks that $a(j+i) \not= a(j+n-l+i)$.  If such an $i$ is found,
it accepts.  We then convert this to a DFA, and interchange accepting
and nonaccepting states.  This DFA $M_1$ accepts
$(j,l,n)$ such that there is no $i$, $0 \leq i < l$ such that
$a(j+i) = a(j+n-l+i)$.   We then use $M_1$ as a subroutine to 
build an NFA $M_2$ that on input $(j,n)$
guesses $l$, checks that $1 \leq l \leq n/2$, and
calls $M_1$ on the result.  We convert this to a DFA and interchange
accepting and nonaccepting states to get $M_3$.  Finally, this
$M_3$ is used as a subroutine to build an NFA $M_4$ that
on input $n$ guesses $j$ and calls $M_3$.  

The set of such integers $n$ then forms a $k$-automatic sequence.
\end{proof}

\begin{example}
\normalfont
Consider the problem of determining for which lengths the Thue-Morse
sequence has an unbordered factor.  Currie and Saari \cite{Currie&Saari:2009}
proved that
if $n \not\equiv \fmod{1} {6}$, then
there is an unbordered factor of length $n$.  (Also see
\cite{Saari:2008}, Lemma 4.10 and Problem 4.1.)  However, this is not
a necessary condition, as 
$${\bf t}[39..69] = 
{\tt 0011010010110100110010110100101},$$
which is an unbordered factor of length $31$.
They
left it as an open problem to give a complete characterization
of the lengths for which $\bf t$ has an unbordered factor.
Our method shows the characteristic sequence of such lengths is
$2$-automatic.

Further, we conjecture that there is an unbordered
factor of length $n$ in $\bf t$ if and only if the base-$2$ expansion of
$n$ (starting with the most significant digit) is not of the
form $1 (01^*0)^* 1 0^* 1$.  

In principle this could be verified, purely mechanically, by our
method, but we have not yet done so.
\end{example}

We now turn to deciding if a given automatic sequence $\bf x$ has
infinite critical exponent (e.g.,~\cite{Krieger&Shallit:2007}).
If a word $w$ can be written in the form $x^n x'$, where $n \geq 1$
is an integer and $x'$ is a prefix of $x$, then we say it is a fractional
power with \textit{exponent} $|w|/|x|$.  For example, ${\tt ingoing}$ has
exponent $7/4$.  The largest such exponent is called \textit{the} exponent of
the word.  The \textit{critical exponent} of $\bf x$ is the supremum, over all
finite factors $f$ of $\bf x$, of the exponent of $f$.

\begin{theorem}
The following question is decidable:  given a $k$-automatic sequence,
does it contain powers of arbitrarily large exponent?
\end{theorem}

\begin{proof}
$\bf x$ has powers of arbitrarily high exponent

\smallskip

iff

\smallskip

\noindent the set of pairs 
$$ S := \lbrace (n,j) \ :  \ \text{ 
$\exists i \geq 0$ such that for all $t$ with $0 \leq t < n$ we have
${\bf x}[i+t] = {\bf x}[i+j+t]$ } \rbrace $$
contains pairs $(n,j)$ with $n/j$ arbitrarily large

\smallskip

iff

\noindent for all $i \geq 0$ $S$ contains a pair $(n,j)$ with
$n > j \cdot 2^i$

\smallskip

iff

\smallskip

\noindent $L$, the set of base-$k$ encodings of pairs in $S$, contains,
for each $i$, 
words ending in 
$$\overbrace{[*,0][*,0]\cdots [*,0]}^i [b,0]$$ 
for some $b \not = 0$, where $*$ means any digit.  

But we can easily decide if a regular language contains words ending
in arbitrarily long words of this form. 
\end{proof}

In a similar fashion we can show

\begin{theorem}
The following question is decidable:  given a $k$-automatic sequence
$\bf x$, does $\bf x$ contain arbitrarily large unbordered factors?
\end{theorem}

Now we turn to questions of recurrence.

An infinite word ${\bf a} = (a(n))_{n \geq 0}$ is said to be 
{\it recurrent} if every factor that occurs at least once
in $\bf a$ occurs infinitely often.
Equivalently,  
a word is recurrent
if and only if for each occurrence of a factor
of $\bf a$, there exists a later occurrence of that factor in $\bf a$.
Equivalently,
for every $n \geq 0$, $r \geq 1$,
there exists $m > n$ such that $a(n+j) = a(m+j)$ for $0 \leq j < r$.

Similarly, an infinite word ${\bf a} = (a(n))_{n \geq 0}$ is said to be 
{\it uniformly recurrent} if every factor that occurs at least once
in $\bf a$ occurs infinitely often, with bounded gaps between consecutive
occurrences.
Equivalently, 
a word 
${\bf a} = (a(n))_{n \geq 0}$ is uniformly recurrent iff
for every $r \geq 1$  there exists $t > 0$ such that for every $n \geq 0$
there exists $m \geq 0$ with $n < m < n+t$ such that $a(n+i) = a(m+i)$
for $0 \leq i < r$.

Thus we recover the following recent result
of Nicolas and Pritykin \cite{Nicolas&Pritykin:2009}:

\begin{theorem}
It is decidable if a $k$-automatic sequence is recurrent or uniformly
recurrent.
\end{theorem}

We now turn to questions of factors shared by two $k$-automatic sequences.
Fagnot \cite{Fagnot:1997a} showed that it is decidable whether two
such sequences ${\bf x} = a(0) a(1) \cdots$ and
${\bf y} = b(0) b(1) \cdots$ have exactly the same set of factors.
This is also decidable by our methods, as follows:

The sequences ${\bf x} = a(0) a(1) \cdots$ and
and ${\bf y} = b(0) b(1) \cdots$ have the same set of factors

\smallskip

iff

\smallskip

\noindent for all $i \geq 0, n \geq 1$ there exists $j \geq 0$ such that
${\bf x}[i..i+n-1] = {\bf y}[j..j+n-1]$

\smallskip

iff

\smallskip

\noindent for all $i \geq 0,  n \geq 1$ there exists $j \geq 0$
such that for all $t, 0 \leq t < n$  we have
$a(i+t) = b(j+t)$.

In a similar fashion, the question of whether the set of factors of
one $k$-automatic word form a subset of the set of factors of 
another $k$-automatic word is decidable.

\section{Enumeration}

We now turn to questions of enumeration.  A typical example of the kind of question we are
interested in is, given an automatic sequence $(a(n))_{n \geq 0}$, how many distinct factors
are there of length $n$?  Our goal in the remainder of this paper
is to show that these kinds of questions often have
a useful answer in terms of \textit{$k$-regular sequences}.
A sequence $(a(n))_{n \geq 0}$ is {\it $k$-regular} if the module generated by its $k$-kernel, which is the
set of all subsequences of the form 
$$ \lbrace (a(k^e n + c))_{n \geq 0} \ : \ e \geq 0, \ 0 \leq c < k^e \rbrace,$$
is finitely generated 
\cite{Allouche&Shallit:1992,Allouche&Shallit:2003a,Allouche&Shallit:2003b,Berstel&Reutenauer:2011}.
The $k$-regular
sequences play the same role for integer-valued sequences as the $k$-automatic
sequences play for sequences over a finite alphabet.  Classical examples of
$k$-regular sequences include polynomials in $n$, 
and $s_k (n)$, the sum of the base-$k$ digits of $n$.

Not only does this interpretation give an explicit
and efficient algorithm for computing the values of the sequence in question,
it also gives a way to compute many related quantities that, up to now, have
received extended treatments in the literature using a wide variety of techniques.
Our work therefore extends and unifies many results in the literature.

In order to make our results really precise, we need several sections of 
preliminary definitions and results.  This is what follows in
Sections~\ref{a1}--\ref{b1}.  We resume 
the exposition of our results in Section~\ref{applications}.

\section{$k$-regular sequences}
\label{a1}

Cobham \cite{Cobham:1972} showed that a sequence $(s(n))_{n \geq 0}$
is $k$-automatic iff its \textit{$k$-kernel}
is finite.
Generalizing this notion,
Allouche and Shallit \cite{Allouche&Shallit:1992,Allouche&Shallit:2003a}
introduced the notion of $k$-regular sequence over a ring $R$.  
A sequence is \textit{$k$-regular} if the module generated by its
$k$-kernel is finitely generated.
In particular, Allouche and Shallit were
interested in the cases of where the underlying ring is $\Zee$ or
$\Que$.
However, as noted in the recent book of Berstel and Reutenauer
\cite{Berstel&Reutenauer:2011}, it makes more sense to define the
$k$-regular sequences over a semiring instead of a ring.  The advantage
is greater generality, but at the cost of giving up part of the
characterization in terms of the $k$-kernel.

\begin{example}
\normalfont
To illustrate this, consider the sequence $s_2(n)$ defined to be the sum
of the bits in the base-$2$ representation of $n$.  For example,
$s_2 (27) = 4$.  Then $s_2(n)$ is $2$-regular over $\Zee$, as its 
$2$-kernel $K$ generates a module $M$ that 
is generated by the sequence $s_2(n)$ itself and the 
constant sequence $1$.   Indeed, we have
\begin{eqnarray*}
K &=& \lbrace (s_2(2^e n + a))_{n \geq 0} \ : \ e \geq 0, \ 0 \leq a < 2^e \rbrace \\
&=& \lbrace (s_2 (n) + s_2 (a))_{n \geq 0} \ : \ a \geq 0 \rbrace  \\
&=& \lbrace (s_2(n) + c)_{n \geq 0} \ : \ c \geq 0 \rbrace ,
\end{eqnarray*}
so that every sequence in the $2$-kernel is a $\Zee$-linear combination
of $(s_2(n))_{n \geq 0}$ and the constant sequence $1$.  Indeed, it is
even true that every sequence in $K$ is an $\Enn$-linear combination of
$(s_2(n))_{n \geq 0}$ and the constant sequence $1$.

In \cite{Allouche&Shallit:1992}, the authors show that every sequence in
the $k$-kernel $K$ of a $k$-regular sequence over $\Zee$ 
is generated by some finite subset of $K$.
For example, for $(s_2(n))_{n \geq 0}$, the $2$-kernel $K$ is generated by
$(s_2(n))_{n \geq 0}$ and $(s_2(2n+1))_{n \geq 0}$.
However, in this example,
there is no finite subset $K' \subseteq K$ such that every sequence
in $K$ can be written as an $\Enn$-linear combination of the sequences
in $S$.  For every
sequence in $K$ is of the form $s_2(n) + c$ with $c \geq 0$.  If we
take some finite subset $K' \subseteq K$, then the sequences in $K'$ of
the form $s_2 (n) + c$ all satisfy $c < C$ for some finite $C$.  We
then cannot get $s_2(n) + C+1$ as an $\Enn$-linear combination of the
sequences in $K'$ (as any such combination would have at least two
copies of $s_2 (n)$).
\label{s2}
\end{example}

This means that to define $(\Enn, k)$-regular sequences, we have to
give up one characterization in terms of the kernel, given in
\cite{Allouche&Shallit:1992}.

\section{$(R,k)$-regular sequences}

In this section we give a rigorous definition of $(R,k)$-regular
sequences and show that there are a number of alternative
characterizations that are equivalent.

First, we give some definitions.

Let $\Sigma_k$ denote the alphabet
$\lbrace 0, 1, \ldots, k-1 \rbrace$.
Let $C_k = \lbrace \epsilon \rbrace \ \cup \ 
(\Sigma_k - \lbrace 0 \rbrace)\Sigma_k^*$ denote the set of
canonical base-$k$ expansions, that is, those with no leading zero.
Let $R$ be a semiring.
A {\it formal series} 
is a map 
$h:\Sigma^* \rightarrow R$.  For historical reasons, 
$h(w)$ is often written as $(h,w)$ and $h$ itself is expressed as
the formal sum $\sum_{w \in \Sigma^*} (h,w) w$.
A formal series $h$ taking values in a semiring $R$
is said to be {\it $R$-recognizable} if 
$(h, w) = u \mu(w) v$ for all $w \in \Sigma^*$, where 
$\mu$ is a morphism
from $\Sigma^*$ to the set of $n \times n$ matrices, $u$ is a $1 \times n$
matrix (or row vector), and $v$ is an $n \times 1$ matrix (or column vector),
all with entries in $R$.
The triple $(u, \mu, v)$ is called a 
\textit{linear representation} of $h$.

The reader is directed to \cite{Kuich&Salomaa:1986,Salomaa&Soittola:1978} and
especially \cite{Berstel&Reutenauer:2011} for more information about
recognizable series.

We recall the following standard result about recognizable
series (\cite{Berstel&Reutenauer:2011}, Ex.\ 2.1.3, p.\ 42):

\begin{lemma}
Let $R$ be a semiring,
and let $f:\Sigma_k^* \rightarrow R$ be an $R$-recognizable series.
Then the series $g:\Sigma_k^* \rightarrow R$ defined by
$(g,w) = (f,w^R)$ is also $R$-recognizable.
\label{rev}
\end{lemma}

Next, we prove a somewhat technical lemma that essentially
says that we can 
disregard leading $0$'s in the representation of a word.

\begin{lemma}
Let $R$ be a semiring, and
let $f:\Sigma_k^* \rightarrow R$ be an $R$-recognizable series.
Then there exists another $R$-recognizable series $g$ such that
$(g, 0^i w) = (f,w)$ for all $i \geq 0$ and all $w \in C_k$.
Furthermore, there exists a linear representation
$(u', \mu', v')$ for $g$ satisfying $u'\mu'(0) = u'$.
\label{lz}
\end{lemma}

\begin{proof}
Suppose $(u, \mu, v)$ is a rank-$n$ linear representation of $f$.
Let $I_n$ denote the $n \times n$ identity matrix.
Define $u', \mu', v'$ as follows:
\begin{eqnarray*}
u' &=& [ \ \overbrace{0 \ 0 \cdots \ 0}^n  \quad u ] \\
\mu'(a) &=& \begin{cases}
	\left[ \begin{array}{cc}
	\mu(0) & {\bf 0} \\
	{\bf 0} & I_n 
	\end{array} \right], & \text{ if $a = 0$;}  \\
	\vphantom{a} & \\
	\left[ \begin{array}{cc}
	\mu(a) & {\bf 0} \\
	\mu(a) & {\bf 0} 
	\end{array} \right], & \text{ if $a \not= 0$;} \\
	\end{cases} \\
v' &=& [ v \quad v ]^T, \\
\end{eqnarray*}
and set $g = (u', \mu', v')$.

To see that this works, we will first prove the following two
facts:
\begin{equation}
\mu'(0^i) = \left[ \begin{array}{cc}
	\mu(0^i) & {\bf 0} \\
	{\bf 0} & I_n
	\end{array}\right]
\label{eq1}
\end{equation}
for $i \geq 0$ and
\begin{equation}
\mu'(0^i w) = \left[ \begin{array}{cc}
	\mu(0^i w) & {\bf 0} \\
	\mu(w) & {\bf 0} 
	\end{array}
	\right]
\label{eq2}
\end{equation}
for $i \geq 0$ and $w \in (\Sigma_k - \lbrace 0 \rbrace)\Sigma_k^*$.
The claim (\ref{eq1}) is a trivial induction, and is omitted.
Let's prove (\ref{eq2}) by induction on $|w|$.  The base case
is $|w|=1$.
In that case $w = a$, where $a \in \Sigma_k - \lbrace 0 \rbrace$.  From
the definition we have
$$ \mu'(w) = \mu'(a) =  \left[ \begin{array}{cc}
	\mu(a) & {\bf 0} \\
	\mu(a) & {\bf 0}
	\end{array} \right]
$$ 
so, using (\ref{eq1}), we get
$$ \mu'(0^i w) = \mu'(0^i) \mu'(w) = 
\left[ \begin{array}{cc}
	\mu(0^i)\mu(a) & {\bf 0} \\
	\mu(a) & {\bf 0} \\
	\end{array} \right]
= \left[ \begin{array}{cc}
	\mu(0^i a) & {\bf 0} \\
	\mu(a) & {\bf 0} \\
	\end{array} \right],
$$
as desired.
For the induction step, assume the result (\ref{eq2}) holds for all
$w'$ with $0 < |w'| < |w|$; we prove it for $w$.
Write $w = ax$ with $a \in \Sigma_k - \lbrace 0 \rbrace$.
There are two cases:  (i) $x = 0^j$ for some $j \geq 1$,
and (ii) $x = 0^j y$, where $j \geq 0$ and $y \in C_k$.  
In case (i) we have, by induction, that
$$ \mu'(x) = \left[ \begin{array}{cc}
	\mu(0^j) & {\bf 0} \\
	{\bf 0} & I_n \\
	\end{array} \right],$$
and hence
$$ \mu'(w) = \mu'(a) \mu'(x) =
\left[ \begin{array}{cc}
	\mu(a) \mu(0^j) & {\bf 0} \\
	\mu(a) \mu(0^j) & {\bf 0} \\
	\end{array} \right]
= \left[ \begin{array}{cc}
	\mu(ax) & {\bf 0} \\
	\mu(ax) & {\bf 0} \\
	\end{array} \right] 
= \left[ \begin{array}{cc}
	\mu(w) & {\bf 0} \\
	\mu(w) & {\bf 0} 
	\end{array} \right],
$$ 
as desired.
In case (ii) we have, by induction, that
$$ \mu'(x) =  \mu'(0^j y) = \left[ \begin{array}{cc}
	\mu(0^j y) & {\bf 0} \\
	\mu(y) & {\bf 0}
	\end{array}
	\right],$$
and again we have
$$ \mu'(w) = \mu'(a) \mu'(x) =
\left[ \begin{array}{cc}
	\mu(a) \mu(0^j y) & {\bf 0} \\
	\mu(a) \mu(0^j y ) & {\bf 0} \\
	\end{array} \right]
	= \left[ \begin{array}{cc}
	\mu(ax) & {\bf 0} \\
	\mu(ax) & {\bf 0} \\
	\end{array} \right]
	= 
 \left[ \begin{array}{cc}
	\mu(w) & {\bf 0} \\
	\mu(w) & {\bf 0} 
	\end{array} \right],
$$
as desired.

Therefore
$$ \mu'(0^i w) = \mu'(0^i) \mu'(w) = 
 \left[ \begin{array}{cc}
	 \mu(0^i) & {\bf 0} \\
	{\bf 0} & I_n
	\end{array}\right]
\
 \left[ \begin{array}{cc}
	\mu(w) & {\bf 0} \\
	\mu(w) & {\bf 0} 
	\end{array} \right]
= \left[ \begin{array}{cc}
	\mu(0^i w) & {\bf 0} \\
	\mu(w) & {\bf 0} 
	\end{array}\right],
$$
which completes the induction.

Now that we know that (\ref{eq1}) and (\ref{eq2}) hold, we have, if
$w = 0^i$ for some $i \geq 0$, that
\begin{eqnarray*}
(g,w) &=& u' \mu'(w) v' \\
&=& [ \ \overbrace{0 \ 0 \ \cdots 0}^n \quad u]
\left[ \begin{array}{cc}
	\mu(0^i) & {\bf 0} \\
	{\bf 0} & I_n
\end{array}\right] 
\left[ \begin{array}{c}
	v \\
	\quad \\
	v 
	\end{array}
	\right]\\
&=& uv   \\
&=& (f, \epsilon),
\end{eqnarray*}
as desired.

If $w = 0^i z$ with $i \geq 0$ and $z \in (\Sigma_k - \lbrace 0 \rbrace)
\Sigma_k^*$, then
\begin{eqnarray*}
(g,w) &=& u' \mu' (w) v'  \\
&=& [ \ \overbrace{0 \ 0 \ \cdots 0}^n \quad u] 
\left[ \begin{array}{cc}
	\mu(0^i z) & {\bf 0} \\
	\mu(z) & {\bf 0}
	\end{array}
	\right] 
	\left[ \begin{array}{c}
		v \\
		\quad \\
		v 
		\end{array}
		\right]\\
&=& u \mu(z) v \\
&=& (f, z),
\end{eqnarray*}
as desired.

Finally, note that 
$u' \mu'(0) = [ \ \overbrace{0 \ 0 \ \cdots 0}^n \quad u]
\left[ \begin{array}{cc}
	\mu(0) & {\bf 0} \\
	{\bf 0} & I_n
\end{array} \right] 
= [ \ \overbrace{0 \ 0 \ \cdots 0}^n \quad u] = u'$.
\end{proof}

Combining the previous two lemmas, we get

\begin{lemma}
Let $R$ be a semiring, and let
$f:\Sigma_k^* \rightarrow R$ be an $R$-recognizable series.
Then there exists another $R$-recognizable series $g$ such that
$(g, w 0^i) = (f,w)$ for all $i \geq 0$ and all $w \in C_k^R$.
Furthermore, there exists a linear representation
$(u', \mu', v')$ for $g$ satisfying $\mu'(0) v' = v'$.
\label{tz}
\end{lemma}

We are now ready to state our equivalence theorem.  This result can be
viewed as an expanded version of \cite{Berstel&Reutenauer:2011},
Prop.\ 1.1, p.\ 84.

\begin{theorem}
Let $(f(n))_{n \geq 0}$ be a sequence taking values in a semiring $R$.
The following are equivalent.
\begin{itemize}
\item[(a)] There exist finitely many sequences  $(f_1 (n))_{n \geq 0}, \ldots,
(f_r(n))_{n \geq 0}$
such that 
\begin{itemize}
\item[(i)] 
$(f(n))_{n \geq 0}$ is an $R$-linear
combination of the $f_i$; and
\item[(ii)] 
for each $i$ and $a$ with $1 \leq i \leq r,$ and $0
\leq a < k$, the subsequence $(f_i (kn+a))_{n \geq 0}$ is an $R$-linear
combination of the $(f_i (n))_{n \geq 0}$.
\end{itemize}

\item[(b)] There exist finitely many sequences 
$(f_1 (n))_{n \geq 0}, \ldots, (f_r (n))_{n \geq 0}$ and
$k$ matrices $B_0, B_1, \ldots, B_{k-1}$ with entries in $R$ such
that if $$V(n) = \left( \begin{array}{c} f_1 (n) \\ \vdots \\ f_r(n)
\end{array} \right),$$ then $V(kn+a) = B_a V(n)$ for $0 \leq a < k$.
and there exists a vector $z \in R^{1 \times r}$
such that $f(n) = z \cdot V(n)$.

\item[(c)] There exist a matrix-valued morphism $\mu: \Sigma_k^* \rightarrow
 R^{r \times r} $ and vectors $u,v$ with entries in $R$,
such that $\mu(0) v = v$ and $f(n) = u\, \mu(w^R)\, v$ for all $w\in\Sigma_k^*$ with $[w]_k=n$. 

\item[(d)] There exist a matrix-valued morphism
$\rho: \Sigma_k^* \rightarrow R^{s \times s}$
and vectors $u', v'$ with entries in $R$, such that $u' =
u' \rho(0)$ and $f(n) = u' \rho(w) v'$ for all $w\in\Sigma_k^*$ with $[w]_k = n$.

\item[(e)] There is an $R$-recognizable series $d$ such that $(d,w) =
f([w]_k)$ for all $w \in C_k$.

\item[(f)] The mapping $(h,w) := f([w]_k)$ defines an $R$-recognizable
series. 

\item[(g)] The mapping $(h',w) := f([w^R]_k)$ defines an $R$-recognizable series. 

\item[(h)] There is an $R$-recognizable series $p$ such that $(p,w) =
f([w^R]_k)$  for all $w \in C_k^R$.

\end{itemize}
\label{equiv}
\end{theorem}

\begin{proof}
(a) $\implies$ (b):    Since each $f_i(kn+a)$ is an $R$-linear combination of
the $f_i$, we can express this as the matrix product
$V(kn+a) = B_a V(n)$. Since $f(n)$ is an $R$-linear combination of
the $f_i(n)$, we can express this as $f(n) = z \cdot V(n)$ for a
suitable vector $z$.

\bigskip

(b) $\implies$ (c):    In fact we can take $v = V(0)$, $\mu(a) = B_a$
for $0 \leq a < k$,
and $u = z$.  Let us prove
by induction on $n$ that $V(n) = \mu((n)_k^R) V(0)$.    The base case 
is $n = 0$.  Then $(n)_k = \epsilon$, so $\mu((n)_k^R) = I$, the identity
matrix, and $V(0)= I \cdot V(0)$.  

Now assume the result is true for all $n' < n$, and we prove it for $n$.
Write $n = kn' + a$ for $0 \leq a < k$.  Then by induction
$V(n') = \mu((n')_k^R) V(0)$.   Then $V(n) = V(kn'+a) = \mu(a) V(n') =
\mu(a) \mu((n')_k^R)  V(0) = \mu((n)_k^R) V(0)$.  

We have $f(n) = z V(n)$.  Furthermore, from $V(kn+a) = \mu(a) V(n)$
with $k = 0, n = 0, a = 0$, we get $v = \mu(0) v$.

Finally, if $w\in\Sigma_k^*$ is such that $[w]_k=n$, then $w^R=(n)_k^R0^i$ for some $i\ge0$.
Because $v = \mu(0)^i v$, we have $u \mu(w^R)v=u\mu((n)_k^R)\mu(0)^iv=u\mu((n)_k^R)v=f(n)$.

\bigskip

(c) $\implies$ (d):  Let $\rho(i) := \mu(i)^T$, $u' := v^T$, and 
$v' := u^T$.    Then
from (c) we get $f(n) = u' \rho((n)_k) v'$.    Furthermore, from $v = 
\mu(0) v$
we get $v^T = v^T \mu(0)^T$, and so 
\begin{equation}
u' = u' \rho(0).
\label{star}
\end{equation}
Let $w$ be any
word such that $[w]_k = n$.  Then we can write $w = 0^i (n)_k$ for some
$i \geq 0$.  Then $u' = u' \rho(0)^i$ from (\ref{star}), and hence
$f(n) = u' \rho(w) v'$, as desired.

\bigskip

(d) $\implies$ (e): We can take $d = (u', \rho, v')$.  

\bigskip

(e) $\implies$ (f):  
Let $d$ be an $R$-recognizable series such that $(d,w) =
f([w]_k)$ for all words $w \in C_k$.  Now apply Lemma~\ref{lz};
we obtain a new $R$-recognizable series $h$ with
$(h,0^i x) = (d,x)$ for all $i\geq 0$ and all $x \in C_k$.  
Let $w \in \Sigma_k^*$.  Then we can write $w = 0^j x$, where
$x \in C_k$.  Then $(h,w) = (h,0^j x) = (d,x) = f([x]_k) = f([0^j x]_k) = f([w]_k)$.  

\bigskip

(f) $\implies$ (g):
Using Lemma~\ref{rev}, if $h'$ is the series defined by
$(h', w) = (h, w^R)$, then $h'$ is also $R$-recognizable.
We have $(h',w) = (h, w^R) = f([w^R]_k)$.  

\bigskip

(g) $\implies$ (h): Trivial.

\bigskip

(h) $\implies$ (a): By Lemma~\ref{tz} there exists an $R$-recognizable series
$p'$ with linear representation $(c', \gamma', d')$
such that $(p', w0^i) = (p,w)$ for all $i \geq 0$ and
all $w \in C_k^R$.  Furthermore, $\gamma'(0) d' = d'$.

Define
the sequences $(f_i(n))_{n \geq 0}$ as follows:
$$ \left[ \begin{array}{c} f_1 (n) \\
f_2(n) \\
\vdots \\
f_t(n)
\end{array} \right] = \gamma'( (n)_k^R ) \cdot d'.$$
Then
\begin{equation}
\left[ \begin{array}{c} f_1 (kn+a) \\
f_2(kn+a) \\
\vdots \\
f_t(kn+a)
\end{array} \right] = \gamma'( (kn+a)_k^R ) \cdot d' .
\label{an00}
\end{equation}

If $(a,n) \not= (0,0)$ then 
\begin{eqnarray*}
\gamma'( (kn+a)_k^R ) \cdot d'  &=& \gamma'( a \cdot (n)_k^R ) \cdot d' \\
&=& \gamma'(a) \gamma'((n)_k^R) \cdot d' \\
&=& \gamma'(a) \left[ \begin{array}{c} f_1 (n) \\
f_2(n) \\
\vdots \\
f_t(n)
\end{array} \right],
\end{eqnarray*}
which expresses each $f_i(kn+a)$ as a linear combination of 
$f_1 (n), f_2 (n), \ldots, f_t (n)$.

If $(a,n) = (0,0)$,
then from (\ref{an00}) we get 
$$ 
\left[ \begin{array}{c} f_1 (kn+a) \\
f_2(kn+a) \\
\vdots \\
f_t(kn+a)
\end{array} \right] = \gamma'( (0)_k^R ) \cdot d' =
\gamma'(\epsilon) \cdot d' = d' = \gamma'(0) \cdot d' = 
\gamma'(a) \left[\begin{array}{c} f_1 (n) \\
f_2(n) \\
\vdots \\
f_t(n)
\end{array} \right].
$$

Furthermore,
\[f(n) = (p,(n)_k^R) \\
= c' \cdot \gamma'((n)_k^R) \cdot d' \\
= c' \cdot  \left[ \begin{array}{c} f_1 (n) \\
f_2(n) \\
\vdots \\
f_t(n)
\end{array} \right],\]
which expresses $f(n)$ as a linear combination of
$f_1 (n), f_2 (n), \ldots, f_t (n)$.

\end{proof}

We can now formally define $k$-regular sequences over a semiring.

\begin{definition}
\normalfont
Suppose $R$ is a semiring, and $f:\Enn \rightarrow R$ is a 
sequence with values in $R$.  If any of the conditions
(a)--(h) in Theorem~\ref{equiv} hold, then we say that
$f$ is $(R, k)$-regular.
\end{definition}

\begin{corollary}
If $f$ is a sequence such that
$f(n)$ is an $R$-linear combination of some finite subset of its
$k$-kernel, then it is $(R,k)$-regular.
\label{ker}
\end{corollary}

\begin{proof}
Follows from Theorem~\ref{equiv} (a).
\end{proof}

However, unlike the case of $(\Zee,k)$- or $(\Que, k)$-regular
sequences, the converse to Corollary~\ref{ker} does not hold,
as we have seen above in Example~\ref{s2}.

\section{$\Enn$-recognizable series}
\label{enn}

In this section, our semiring is $R = \Enn$, the non-negative integers.
We prove a characterization of $\Enn$-recognizable series in terms
of automata and transducers
(Theorem~\ref{one}, below).

We recall the notion of nondeterministic finite automaton (NFA):  it is
a 5-tuple $M= (Q,\Sigma, \delta, q_0, F)$, where $Q$ is a set of states,
$\Sigma$ is a finite alphabet, $q_0$ is the initial state, and $F \subseteq Q$
is the set of final states, and $\delta:Q \times \Sigma \rightarrow 2^Q$
is the transition function, extended to $Q \times \Sigma^*$ in the
obvious way.  A \textit{path labeled $w = a_1 \cdots a_n$}
in an NFA is a sequence of states $(p_0, p_1, \ldots, p_n)$
such that $p_{i+1} \in \delta(p_i, a_{i+1})$ for $0 \leq i < n$.
It is an \textit{accepting path} if $p_0 = q_0$ and $p_n \in F$.

We will also be concerned with {\it nondeterministic
uniform finite-state transducers}. 
Such a transducer produces an output of the same length for every
input symbol.  Formally, such a transducer $T = (Q, \Sigma, \Delta,
E, q_0, F)$, where $E \subseteq Q \times \Sigma \times \Delta^l \times Q$
is the set of permissible transitions.  A transition $(q_i, a, y, q_j)$
means that if the transition is in state $q_i$ then on input $a$ it has
the option (nondeterministically) of outputting $y$ and entering state $q_j$.
The output of $T$ on input $w$ is the set of all words formed by
concatenating the outputs on a path labeled $w$
from $q_0$ to some state of $F$.

Finally, given words $w = a_1 a_2 \cdots a_n$ and
$x = b_1 b_2 \cdots b_n $ of the same length, but defined
over possibly different alphabets (say, $\Sigma$ and $\Delta$, respectively),
we define the word $w \times x$ to be the word 
$z = [a_1, b_1] [a_2, b_2] \cdots [a_n, b_n]$ over the alphabet
$\Sigma \times \Delta$ .   In this case, we define the projection maps
$\pi_1(z) = w$ and $\pi_2(z) = x$.

\begin{theorem}
Let $f: \Sigma^* \rightarrow \Enn$ be a formal series with $(f,\epsilon) = 0$. Then the following are equivalent.

\begin{itemize}

\item[(a)] $f$ is $\Enn$-recognizable.

\item[(b)] There exists an NFA $M = (Q, \Sigma, \delta, q_0, F)$ such that
for all $w \in \Sigma^*$, there are exactly $(f,w)$ paths labeled $w$
from $q_0$ to a state of $F$.

\item[(c)] There is an alphabet $\Delta$ and a regular language
$L \subseteq (\Sigma \times \Delta)^*$ such that 
$$(f,w) = \left| \lbrace z \in L \ : \ \pi_1(z) = w \rbrace \right|$$
for all words $w$.

\item[(d)] There is an alphabet $\Delta$ and a
 (nondeterministic) $1$-uniform
finite-state
transducer $T:\Sigma^* \rightarrow \Delta^*$ such that $(f,w) = |T(w)|$ 
for all words $w$.

\end{itemize}
\label{one}
\end{theorem}

\begin{proof}
\noindent (a) $\implies$ (b):  Since $f$ is recognizable, there is
a matrix representation $(u, \mu, v)$ such that
$(f,w) = u \mu(w) v$ for all $w \in \Sigma^*$.  
By an
exercise in \cite{Sakarovitch:2009}, Ex.\ III.3.3, p.\ 426,
we can, without loss of generality, assume that
$u = [1\ \overbrace{0\ 0\ \cdots\ 0}^{n-1}]$ and
$v = [\overbrace{0\ 0\ \cdots\ 0}^{n-1}\ 1]^T$ for some $n \geq 2$.
For completeness, we give the proof here:

Given a rank-$t$ representation
$(u, \mu, v)$, we produce a new rank-($t+2$) representation
$(u', \mu', v')$ defined as follows:

\begin{eqnarray*}
u' &=& [1\ \overbrace{0\ 0\ \cdots\ 0}^{t+1}] \\
\mu'(a) &=& \left[ 
	\begin{array}{ccc}
	0 & [ u \cdot \mu(a) ] & [u \cdot \mu(a) \cdot v] \\
	\begin{array}{c}
	0 \\
	\vdots \\
	0
	\end{array} & 
	\left[
	\begin{array}{c}
	\\
	\mu(a) \\
	\\
	\end{array} \right]
	&
	\left[ \begin{array}{c}
	\\ 
	\mu(a) \cdot v \\
	\\
	\end{array} \right] \\
	0 & 0 \cdots 0 & 0 
	\end{array} \right] \\
v' &=& [\overbrace{0\ 0\ \cdots\ 0}^{t+1}\ 1]^T .
\end{eqnarray*}

Now an easy induction on $|w|$ shows that, for $|w| \geq 1$,
that
$$
\mu'(w) =\left[ 
	\begin{array}{ccc}
	0 & [ u \cdot \mu(w) ] & [u \cdot \mu(w) \cdot v] \\
	\begin{array}{c}
	0 \\
	\vdots \\
	0
	\end{array} & 
	\left[
	\begin{array}{c}
	\\
	\mu(w) \\
	\\
	\end{array} \right]
	&
	\left[ \begin{array}{c}
	\\ 
	\mu(w) \cdot v \\
	\\
	\end{array} \right] \\
	0 & 0 \cdots 0 & 0 
	\end{array} \right] .
$$
It follows that, for $w \not= \epsilon$, 
$u' \mu'(w) v' = u \mu(w) v$.  For $w = \epsilon$, we have
$u' \mu'(w) v' = 0$.  This completes the proof of the exercise.

Now that this somewhat technical point has been handled, we turn to the
idea of the construction.  The desired interpretation is that
$\mu(w)_{i,j}$ should count the number of paths labeled $w$
from state $i$ to state $j$.  However, this is not sensible if
$a$ is a single symbol, as there is either one directed edge
in the automaton from $i$ to $j$ labeled $a$, or none.  To get around
this problem, we make multiple copies of each state, and create a transition
from $i$ to $\mu(w)_{i,j}$ copies of state $j$.

    From the rank-$n$ linear representation for $f$, namely
$(u, \mu, v)$, we create an NFA
$M = (Q, \Sigma, \delta, q_0, F)$ with $(f,w)$ paths labeled $w$.
Let $m$ be the maximum entry in all the $\mu(a)$, $a \in\Sigma$.
Define
\begin{eqnarray*}
Q &=& \lbrace [i,j] \ :\ 1 \leq i \leq n,\ 1 \leq j \leq m \rbrace \\
q_0 &=& [1,1] \\
F &=& \lbrace [n, s] \ : \ 1 \leq s \leq m \rbrace \\
\delta([i,j], a) &=& \lbrace [r,s] \ : \ 1 \leq r \leq n, \ 
	1 \leq s \leq \mu(a)_{i,r} \rbrace,
\end{eqnarray*}
where by $\mu(a)_{i,r}$ we mean the entry in row $i$ and column $r$ of
the matrix $\mu(a)$.

To see that this works, let $P_{i,j,r} (w)$ denote the number of paths
labeled $w$
from $[i,j]$ to some member of $\lbrace [r,s] \ : \ 1 \leq s \leq m \rbrace$.
We claim that 
\begin{equation}
P_{i,j,r} (w) = \mu(w)_{i,r}
\label{pmu}
\end{equation}
for all $i, j, r, w$ such that
$1 \leq i,r \leq n$, $1 \leq j \leq m$, and $w \in \Sigma^*$.

The proof is by induction on $|w|$.  The base case is $|w| = 0$.  
In this case 
$$\mu(w)_{i,r} =  \begin{cases}
	1, & \text{if $i = r$}; \\
	0, & \text{otherwise.}
	\end{cases}
$$
and the only path of length $0$ from state $[i,j]$ is to $[i,j]$ itself,
so $P_{i,j,r} (w) = \mu(w)_{i,r}$.

Now assume (\ref{pmu}) holds for all $|w'| < |w|$; we prove it
for $w$.  Write $w = ax$ with $a \in \Sigma$.
Break the path labeled $w$ into two pieces, one labeled
$a$ and the other labeled $x$.  Then
\begin{eqnarray*}
P_{i,j,r} (ax) &=& \sum_{{i', j':} \atop
	{[i', j'] \in \delta([i,j], a) }} P_{i', j', r} (x) \\
&=& \sum_{{i',j' :} \atop
	{[i', j'] \in \delta([i,j], a) }} \mu(x)_{i', r} \quad \quad
		\text{(by induction)}\\
&=& \sum_{i'} \mu(a)_{i, i'} \cdot \mu(x)_{i', r} \\
&=& \mu(ax)_{i,r}, 
\end{eqnarray*}
which completes the induction.

Thus $u \cdot \mu(w) \cdot v = \mu(w)_{1,n} = P_{1,1,n} (w)$, as desired.

\bigskip

\noindent (b) $\implies$ (c):
Given the NFA $M = (Q, \Sigma, \delta, q_0, F)$, we
take $\Delta = Q$ and define
$$L = \lbrace w \times x \ : \ w \in \Sigma^*, \ x \in \Delta^*
\text{ such that } q_0 x \text{ is an accepting path for } w \rbrace.$$
Since there are $(f,w)$ accepting paths for $w$ in $M$, and each
such path begins with $q_0$, the result follows.  Clearly $L$ is regular,
as it can be accepted by a simple modification of $M$.

\bigskip

\noindent (c) $\implies$ (d):
Consider a DFA $M$ accepting $L$.  We construct a transducer $T$
with the same set of states, 
initial state, and 
set of final states as $M$.
For each transition in $M$ of the form
$\delta(q_i, [a,b]) = q_j$, we define a transition in $T$ 
from $q_i$ to $q_j$ with input $a$ and output $b$.
It follows that on input $w$, the transducer $T$ outputs all those 
$x$ of the same length for which $w \times x \in L$.

\bigskip

\noindent (d) $\implies$ (a):  Given such a transducer $T = (Q, \Sigma, 
\Delta, E, q_1, F)$, we define the matrix representation $(u, \mu, v)$ for the series $f$
as follows:  if $Q = \lbrace q_1, \ldots, q_n \rbrace$, 
then $u = [1\ \overbrace{0\ 0\ \cdots\ 0}^{n-1}]$ and
$v$ has a $1$ in the entries corresponding to final states
of $F$, and $0$ elsewhere.  Since $|T(\epsilon)| = 0$
by hypothesis, it must be that $q_1 \not\in F$.
Now define $\mu(a)_{i,j}$ to be the number of symbols
$b$ such that $(q_i, a, b, q_j) \in E$.  

\end{proof}

\begin{openproblem}
\normalfont
The preceding theorem would be true if we replace the $1$-uniform
finite-state transducer with any transducer where no two different
paths labeled $w$ give the same output.  One way to ensure this is
that the output labels form a code, and this is clearly true if the
outputs are all of the same length.  What happens if the output
labels do not form a code?  Is the result still true?
\end{openproblem}

\begin{remark}
\normalfont
Carpi and Maggi \cite{Carpi&Maggi:2001}
defined the class of $k$-synchronized sequences, a class
which contains the $k$-automatic sequences and is properly contained in
the class of $k$-regular sequences.  A sequence $(u_n)_{n \geq 0}$ is
$k$-synchronized if the relation $ \{( (n)_k, (u_n)_k ) : n \geq 0 \}$ is a
right-synchronized rational relation.  Roughly speaking, this means
that the relation is realized by a length-preserving rational
transduction, except that we also permit the presence of ``padding''
symbols at the end of one or the other component of the input.  Our
transducer-based characterization, combined with Theorem~\ref{equiv}, characterizes the
more general class of $k$-regular sequences.
\end{remark}

In the usual case where $|\Sigma| \geq 2$, we can take $\Delta$ to be
$\Sigma^l$ for a suitable $l$, as the following theorem shows.

\begin{theorem}
Let $f: \Sigma^* \rightarrow \Enn$ be a formal series with
$|\Sigma| \geq 2$ and $(f,\epsilon) = 0$. Then the following are equivalent.

\begin{itemize}
\item[(a)] $f$ is $\Enn$-recognizable.

\item[(b)] There is an integer $l \geq 1$ and a regular language 
$L \subseteq (\Sigma \times \Sigma^l)^*$ such that 
$$(f,w) = \left| \lbrace z \in L \ : \ \pi_1(z) = w \rbrace \right|.$$
for all words $w$.

\item[(c)] There is an integer $l \geq 1$ and a
 (nondeterministic) $l$-uniform
finite-state
transducer $T:\Sigma^* \rightarrow \Sigma^*$ such that $(f,w) = |T(w)|$ .

\end{itemize}
\label{two}
\end{theorem}

\begin{proof}
Just like the proof of Theorem~\ref{one}.  The only difference is that
we need to choose 
$l$ large enough
so that $|\Sigma|^l \geq |\Delta|$; then we just use 
elements of $\Sigma^l$ instead of those in $\Delta$.
\end{proof}

\section{$\Enni$-recognizable series}

In this section,
we consider the case where the underlying semiring is
$R = \Enni$, the extended non-negative integers.  Roughly speaking, this extension
corresponds to the case where a nondeterministic finite automaton or transducer
is extended by allowing $\epsilon$-transitions.

In addition to the usual
interpretation for addition and multiplication of natural numbers,
we need the following additional rules that turn $\Enni$ into a semiring:
\begin{itemize}
\item[(i)] $a + \infty = \infty + a = \infty$ for all $a \in \Enn_{\infty}$; 
\item[(ii)] $a \cdot \infty = \infty \cdot a = \infty$ for all $a \not = 0$;
\item[(iii)] $0 \cdot \infty = \infty \cdot 0 = 0$.
\end{itemize}
Matrices and vectors with entries in $\Enni$ can now be multiplied using
the usual rules for such multiplication, in addition to the rules
(i)--(iii), as needed.

Let $L$ be a regular
language.  The \textit{characteristic series} of $L$,
denoted $\chi_L$, is the formal series such that
$(\chi_L, w) = 1$ if $w \in L$ and $0$ otherwise.
The {\it Hadamard product} of two series, $h \had h'$, is the
term-by-term product, $(h\had h')(w) = h(w) h'(w)$.  The essential
lemma is the following:

\begin{lemma}
Given a recognizable formal series $f$ over $\Enni$, we can express
it as 
$$f = \chi_{\overline{L}} \had g + \chi_L \cdot \infty,$$
where $g$ is a recognizable formal 
series over $\Enn$ and  $L$ is a regular language.
Furthermore, in the sum, we never add $\infty$ to a value other than $0$.
\label{twoe}
\end{lemma}

\begin{proof}
There are two main ideas.  The first is that the language
$L = \lbrace w \ : \ (f,w) = \infty \rbrace$ is regular.
The second is that $g$ can be taken
to be a modification of $f$ with all occurrences of $\infty$ removed.

First, the construction of $L$.  This is essentially that given in
Salomaa and Soittola \cite{Salomaa&Soittola:1978}, p.\ 40, Exercise 5.

We create a new finite semiring $R' = \lbrace 0, p, \infty
\rbrace$ where the addition and multiplication rules are given as
follows:

\begin{displaymath}
\begin{array}{c|ccc}
+  & 0 & p & \infty \\
\hline
0 & 0 & p & \infty \\
p & p & p & \infty \\
\infty & \infty & \infty & \infty \\
\end{array}
\quad\quad\quad
\begin{array}{c|ccc}
\cdot  & 0 & p & \infty \\
\hline
0 & 0 & 0 & 0 \\
p & 0 & p & \infty \\
\infty & 0 & \infty & \infty \\
\end{array}
\end{displaymath}
Here $0, p, \infty$ should be treated as formal symbols, but the
intent is that the symbol $p$ denotes ``some positive integer''.
Now we define a morphism of semirings $\Enni \rightarrow R'$ as follows:
\begin{eqnarray*}
\tau(0) &=& 0 \\
\tau(i) &=& p, \quad \text{for $0 < i < \infty$} \\
\tau(\infty) &=& \infty .
\end{eqnarray*}
It is now easy to check that for all $a, b \in \Enni$ we have
$\tau(ab) = \tau(a)\cdot \tau(b)$ and
$\tau(a+b) = \tau(a) + \tau(b)$, where the operations on the right-hand-side
are those in $R'$.  We extend $\tau$ to apply to vectors and matrices
by applying $\tau$ to each entry.

Next, we consider the formal series
$f' := \tau\circ f$, which takes its values in
$R'$.  It follows from above that $f'(w) = \hat{u} \hat{\mu} (w) \hat{v}$,
where $\hat{u} = \tau(u)$, $\hat{v} = \tau(v)$, and $\hat{\mu}
= \tau \circ \mu$.  

Now we can create a deterministic finite automaton 
$M = (Q, \Sigma, \delta, q_0, F)$ that essentially computes
the series $f'$.  We do this by letting $Q$ be the set of all possible
$1 \times n$ row vectors over $R'$, letting $q_0 = \hat{u}$, and defining
the transitions $\delta(q, a) = q \cdot \hat{\mu}(a)$.   If we
define $\varphi(t) = t \cdot \hat{v}$, then an easy induction
gives that $\delta(q_0, w) = \hat{u} \cdot \hat{\mu}(w)$, and hence
$\varphi(\delta(q_0,w)) = \hat{u} \cdot \hat{\mu}(w) \cdot \hat{v}
= f'(w)$, as
desired.

We can now define $L$.  Let $F$, the set of final states of $M$,
be given by
$$ F = \lbrace t \in Q \ : \ t \cdot \hat{v} = \infty \rbrace.$$
Then $L = L(M)$. 
By construction, we have the following equivalences:
$w\in L\Leftrightarrow (f',w)=\infty \Leftrightarrow (f,w)=\infty$.

Now we turn to the construction of $g$.
Let $(u, \mu, v)$ be a linear representation for $f$.   
Define a map $\xi: \Enni \rightarrow \Enn$ as follows:
$$
\xi(i) = \begin{cases}
	i, & \text{if $i \in \Enn$}; \\
	0, & \text{if $i = \infty$},
\end{cases}
$$
and extend $\xi$ to apply element-by-element to vectors and matrices in
the obvious way.  Let $g=(u', \mu', v')$, where $u' = \xi(u)$,
$\mu'(a) = \xi(\mu(a))$ for each $a \in \Sigma$, and $v' = \xi(v)$.  
The series $g$ is created by replacing each
occurrence of $\infty$ in $u, \mu, $ and $v$ with $0$.
Then $g$ is evidently
$\Enn$-recognizable, and we claim that $(f,w) \not= \infty \implies (f,w) = (g,w)$.
To see this, note that if $(f,w) \not= \infty$, then in the calculation
$u \cdot \mu(w) \cdot v$ any occurrences of $\infty$ that arise must eventually
be multiplied by $0$, yielding $0$. Then replacing $\infty$ with $0$ has
no effect, since any multiplication involving $0$ will also yield $0$.
(Note that we are not claiming anything about those $w$ for
which $(f,w) = \infty$; the corresponding values of $g$ could be anything.)

It now follows that $f = \chi_{\overline{L}} \had g + \chi_L \cdot \infty$.
\end{proof}

We get the following two corollaries.

\begin{corollary}\label{cor}
If $f : \Sigma^* \to \Enn$ is an $\Enni$-recognizable series,
then it is $\Enn$-recognizable.
\end{corollary}

\begin{proof}
   From Lemma~\ref{twoe},
we have $f= \chi_{\overline{L}} \had g + \chi_L \cdot \infty,$
where $g$ is an $\Enn$-recognizable formal series and $L$ is a regular language.
Since $f(\Sigma^*)\subseteq \Enn$ and from the proof of Lemma~\ref{twoe},
we may choose $L=\emptyset$. 
So $f=g$.
\end{proof}

\begin{corollary}
Given a recognizable formal series $f$ over $\Enni$,
with linear representation $(u, \mu, v)$,
there exists another linear representation $(p, \beta, q)$ such that
the only entries equal to $\infty$ lie in $p$.  
\label{onee}
\end{corollary}

\begin{proof}
First, by a well-known result on the Hadamard product
(e.g., \cite{Schutzenberger:1962} and \cite{Berstel&Reutenauer:2011}, p.\ 15,
the formal series $g' := \chi_{\overline{L}} \had g$ constructed 
in Lemma~\ref{twoe} is recognizable (over $\Enn$).    So $g'$ has a linear
representation $(u, \mu, v)$ that contains no entries of $\infty$.  
Similarly, since $L$ is a regular language,
the characteristic series $\chi_L$ has a linear representation
$(r, \alpha, s)$, where the entries of $r, \alpha, $ and $s$ are all
either $0$ or $1$.  From this we can form a new linear representation
$(p, \beta, q)$ for $f$, via a direct sum construction, as follows:
\begin{eqnarray*}
p &=& \, [ \ u \ \quad \ r\cdot \infty ] \\ 
\beta(a) &=& \left[ \begin{array}{cc} 
	\mu(a) & {\bf 0} \\
	{\bf 0} & \alpha(a) 
	\end{array} \right] \\
q &=& \, [ \ v \ \quad \ s ]^T .
\end{eqnarray*}
Here, $\bf 0$ represents a matrix of $0$'s of the appropriate size.

A routine induction shows that $\beta(w)$ contains
$\mu(w)$ in the upper left and $\alpha(w)$ in the lower right,
from which the result follows.
\end{proof}

\section{Characterizations of $\Enni$-recognizable series}
\label{b1}

Just as the $\Enn$-recognizable series have a number of
different interpretations in terms of automata and
transducers, as we saw in Section~\ref{enn}, so do the
$\Enni$-recognizable series; the difference is that we need to allow
$\epsilon$-transitions.

\begin{theorem}
Let $f:\Sigma^* \rightarrow \Enni$ be a formal series.  Then
the following are equivalent:

\begin{itemize}
\item[(a)]  $f$ is $\Enni$-recognizable;

\item[(b)] There exists an NFA-$\epsilon$ $M = (Q, \Sigma, \delta, q_0, F)$
such that, for all $w \in \Sigma^*$, there are exactly $(f,w)$ paths 
labeled $w$ from $q_0$ to a state of $F$;

\item[(c)] There is an alphabet $\Delta$,
a symbol ${\tt B} \not\in \Sigma$
and a regular
language $L \subseteq ((\Sigma \ \cup \ \lbrace {\tt B} \rbrace) \times
\Delta)^*$ such that
$$(f,w) = | \lbrace z \in L \ : \  \tau(\pi_1(z)) = w \rbrace |,$$
where $\tau$ is the morphism that maps $a$ to $a$ for $a \in \Sigma$
and $\tt B$ to $\epsilon$;

\item[(d)] There is an alphabet $\Delta$, a symbol ${\tt B}
\not\in \Sigma$, and a regular
language $L \subseteq ((\Sigma \ \cup \ \lbrace {\tt B} \rbrace) \times
\Delta)^*$ such that
$$(f,w) = | \lbrace z \in L \ : \ \pi_1(z) \in w{\tt B}^* \rbrace |.$$

\item[(e)] There is an alphabet $\Delta$ and a nondeterministic
finite-state transducer $T$, with inputs of a single letter or
$\epsilon$ on every transition, and outputs of a single letter on
every transition, such that $(f,w) = |T(w)|$.
\end{itemize}
\label{onen}
\end{theorem}

\begin{proof}
We prove the implications in the order
(a) $\implies$ (d) $\implies$ (e) $\implies$ (c) $\implies$ (b)
$\implies$ (a).

\bigskip

\noindent (a) $\implies$ (d):  By Theorem~\ref{twoe}, we know that
$f = \chi_{\overline{L_1}} \had g + \chi_{L_1} \cdot \infty$, where
$L_1 \subseteq \Sigma^*$
is a regular language and $g$ is an $\Enn$-recognizable series. 
Define $g' := \chi_{\overline{L_1}-\lbrace \epsilon \rbrace} \had g$;
then $g'$ is an $\Enn$-recognizable series with $(g', \epsilon) = 0$, so 
we can apply the implication (a) $\implies$ (c) in
Theorem~\ref{one}  to $g'$ to get an alphabet $\Delta$
and a regular language
$L_2 \subseteq (\Sigma \times \Delta)^*$ such that 
$$(g',w) = \left| \lbrace z \in L_2 \ : \ \pi_1(z) = w \rbrace \right|$$
for all words $w \not= \epsilon$.
Let $a$ be an arbitrarily chosen, fixed symbol of $\Delta$,
and consider the language $L_3$ defined by
$$L_3 = \left( \bigcup_{0 \leq i < (f,\epsilon)} [{\tt B},a]^i
\right) \ \cup \  \lbrace z \in 
((\Sigma \ \cup \ \lbrace {\tt B} \rbrace) \times
\Delta)^* \ : \ \pi_1(z) \in (L_1 - \lbrace \epsilon \rbrace)
\cdot {\tt B} ^* \text{ and } \pi_2 (z) \in a^*
\rbrace.$$
It is easy to see that $L_3$ is regular, as each term of the big union
is regular. For the second term, we can, given a DFA for $L_1 - \lbrace
\epsilon \rbrace$, modify it by
\begin{itemize}
\item changing each transition on any letter $b$ to a transition on $[b,a]$
\item adding transitions out of each accepting state
on $[{\tt B}, a]$ to a new final state $q$ and
\item adding a self-loop labeled $[{\tt B},a]$ from $q$ to itself.
\end{itemize}
Let $L := 
L_2 \ \cup \ L_3$.  Then $L$ is regular
and, by construction,
$(f,w) = | \lbrace z \in L \ : \ \pi_1(z) \in w {\tt B}^* \rbrace |$.

\bigskip

(d) $\implies$ (e):  Given a DFA $M$ for $L$, say
$M = (Q, \Sigma', \delta, q_0, F)$ 
where $\Sigma' = (\Sigma \ \cup \ \lbrace {\tt B} \rbrace) \times \Delta$, 
we create the transducer
$T$ with the same set of states, 
initial state, and 
set of final states as $M$.  For each transition in $M$
of the form $\delta(q_i,[a,b]) = q_j$, we define a transition
in $T$ from $q_i$ to $q_j$ with input $a$ and output $b$,
except that if $a = {\tt B}$, then we set the corresponding
transition in $T$ to have input $\epsilon$.  
Each word that $M$ accepts, 
having first component $w{\tt B}^*$ and second component
$y$, corresponds to an input $w$ of $T$ and
an output of $y$.  The result now follows.

\bigskip

(e) $\implies$ (c):  The construction of the previous paragraph is
completely reversible,
which shows that (e) $\implies$ (d).  But clearly (d)
$\implies$ (c).

\bigskip

(c) $\implies$ (b):   Let $L \subseteq
((\Sigma \ \cup \ \lbrace {\tt B} \rbrace) \times
\Delta)^*$ be a regular language such that
$(f,w) = | \lbrace z \in L \ : \  \tau(\pi_1(z)) = w \rbrace |,$
and let $M = (Q, \Sigma', \delta, q_0, F) $ be a DFA accepting $L$,
where $\Sigma' = (\Sigma \ \cup \ \lbrace {\tt B} \rbrace) \times
\Delta$.
We now create an NFA-$\epsilon$ $M'$ with the desired property,
by modifying $M$, as follows:  first, the set of states is expanded
from $Q$ to $Q \times \Delta$.
Second, if $M$ has a transition
$\delta(q_i, [a,b]) = q_j$ with $a \in \Sigma$, then
$M'$ has transitions $\delta([q_i,c], a) = [q_j, b]$ for all
$c \in \Delta$.
Similarly, if $M$ has a transition
$\delta(q_i, [{\tt B},b]) = q_j$, then $M'$ has a transition
$\delta([q_i,c], \epsilon) = [q_j, b]$ for all $c \in \Delta$.
The initial state is $[q_0, c]$ for some arbitrary element 
$c\in \Delta$, and the set of final states of $M'$ is $F \times \Delta$.
The formal proof that this works is essentially the proof of 
(a) $\implies$ (b) in Theorem~\ref{one} and is omitted.

\bigskip

(b) $\implies$ (a):  Given the NFA-$\epsilon$ $M = (Q, \Sigma, \delta,q_0, F)$, we create some associated matrices
$D_a$ for $a \in \Sigma \ \cup \ \lbrace \epsilon \rbrace$.
If the set of states $Q = \lbrace q_0, q_1, \ldots, q_{n-1} \rbrace$,
then $D_a$ has a $1$ in row $i$ and column $j$ iff
$\delta(q_i, a) = q_j$.  

Now any finite path labeled $w = a_1 a_2 \cdots a_n$ in the transition
diagram of $M$ 
looks like
$$ \overbrace{\epsilon, \ldots, \epsilon}^{b_0},\ a_1,\ 
\overbrace{\epsilon, \ldots, \epsilon}^{b_1},\ a_2,\ \ldots,\ 
\overbrace{\epsilon, \ldots, \epsilon}^{b_{n-1}},\  a_n, \  
\overbrace{\epsilon, \ldots, \epsilon}^{b_n}  ,$$
for some $b_0, b_1, \ldots, b_n$ with $0 \leq b_i < \infty$.

Let $D = \sum_{i \geq 0} D_{\epsilon}^i$; this is a matrix with possibly
infinite entries.  Then the entry in row $i$ and 
column $j$ of $D D_{a_1} D D_{a_2} \cdots D D_{a_n} D$ gives
the number of paths from state $q_i$ to state $q_j$ in $M$.
If $u = [1 \ 0 \ 0 \ \cdots \ 0]$ and $v$ is the $\lbrace 0,1 \rbrace$-vector
corresponding to the final states of $M$, then
$u D D_{a_1} D D_{a_2} \cdots D D_{a_n} D v$ is the number of
accepting paths labeled $w$.

If we now define $\mu(a) = D D_a$ for $a \in \Sigma$ and
$v' = Dv$, then $(u, \mu, v')$ is a linear representation for $f$. 
\end{proof}

\section{Applications to enumeration}
\label{applications}

Now that all the basic definitions and results are out of the way, we can resume
our work on enumeration.  The common theme in what follows is to show that
some well-studied sequence is $k$-regular, by combining Theorem~\ref{one} or
Theorem~\ref{onen} (which characterize the formal series associated
with counting the number of
paths, or certain subsets of regular languages, or size of transduced sets, as
$\Enn$- or $\Enni$-recognizable)
with Theorem~\ref{equiv}, which shows the equivalence between $k$-regular
sequence and recognizable formal series.   Here is a simple example:

\begin{theorem}
Let $E$ be any finite set of integers, and consider $(b(n))_{n \geq 0}$, the sequence
that counts the number of 
reversed representations of $n$ in base $k$,
where the digits are chosen only from $E$, and where reversed representations with
trailing zeroes are not allowed.  Then
$(b(n))_{n \geq 0}$ is $(\Enn,k)$-regular.
\label{ns}
\end{theorem}

\begin{proof}
We construct a transducer $T$ having $b(n)$ distinct
outputs on input $(n)_k^R$.   On input $(n)_k$, the transducer $T$ 
guesses a possible representation $w$ using only the digits of $E$,
simultaneously ``normalizes'' it, digit-by-digit, to $w'$, and checks
that the normalized representation is equal to the input.  If it is, then
$w$ is output.  There are some details to handle if $w$ is
shorter or longer than $(n)_k^R$.  If $w$ is shorter, then we allow
padding of $w$ with trailing zeroes.  If $w$ is longer, then 
we handle this by permitting $T$ to perform $\epsilon$-transitions on 
the input after it has processed all the symbols of $(n)_k^R$.  

Then, using Theorem~\ref{onen} together with Theorem~\ref{equiv} and Corollary~\ref{cor}, we see that
$(b(n))_{n \geq 0}$ is an $(\Enn,k)$-regular sequence.
\end{proof}

\begin{example}
\normalfont
Let $b_k (n)$ denote the number of representations of $n$ in base $2$, using
the digits $\lbrace 0, 1, \ldots, k-1 \rbrace$.  Then $b_2 (n) = 1$,
from the uniqueness of binary representations, and $b_3(n)$ is the
Stern-Brocot sequence evaluated at $n+1$.  From Theorem~\ref{ns}, we see that
all these sequences are $(\Enni, 2)$-regular.
See \cite{Reznick:1990}.
\end{example}

We now turn to our main enumeration results.

\begin{theorem}
Let ${\bf x} = a(0) a(1) a(2) \ldots$ be a $k$-automatic sequence.  
Let $b(n)$ be the number of distinct factors of length $n$ in $\bf x$.
Then $(b(n))_{n \geq 0}$ is an $(\Enn,k)$-regular sequence.
\end{theorem}

\begin{proof}
To count distinct factors of length $n$, we count the
\textit{first} occurrences of each factor.

The number of distinct factors of length $n$ in $\bf x$ equals
the number of indices $i$ such that there is no index $j < i$ with
the factor of length $n$ beginning at position $i$ equal to the
factor of length $n$ beginning at position $j$.  

Consider the set 
\begin{eqnarray*}
S & =& \lbrace (n, i) \ : \ \text{for all}\ j \ 
\text{with $0 \leq j < i$ 
there exists an integer }  \\ 
&&  \quad\quad t \text{ with } 0 \leq t < n
\text{ such that } a(i+t) \not= a(j+t) \rbrace .
\end{eqnarray*}
Then, by Theorem~\ref{logic}, the language $S'$ defined to be
the base-$k$ encoding of elements of
$S$, forms a regular language.  We assume without loss of generality that
if one representation of $(n,i)$ appears in $S'$, then they all do, 
including the ones with leading (actually, trailing zeroes).

We now apply a transducer to
$S'$, changing every representation of $(n,i)$ as follows:  we change
every $0$ after the last nonzero digit in the first component to $\tt B$.
This transformation preserves the regularity of $S'$.
Finally, we discard every representation that ends with 
$[{\tt B}, 0]$.   The effect of this is to ensure that $n$ in the
first component, up to ignoring the ${\tt B}$'s, has a single
representation, and that each $i$ corresponding to a particular $n$
has a unique representation.
Using Theorems~\ref{one} and \ref{equiv}, we see that $(b(n))_{n \geq 0}$ is $(\Enn,k)$-regular.
\end{proof}

\begin{remark}
\normalfont
Moss\'e \cite{Mosse:1996a} proved, among other things,
that a sequence that is the fixed point
of a $k$-uniform morphism has a $k$-regular subword complexity function.
With our technique, we obtain her result for these sequences
and also the slightly more general case of $k$-automatic sequence.
\end{remark}

\begin{theorem}
The sequence counting the number of palindromic factors of length $n$ is $(\Enn,k)$-regular.
\end{theorem}

\begin{proof}
The number of distinct palindromes of length $n$ in {\bf x}

\smallskip

is equal to

\smallskip

\noindent the number of indices $i$ such that
${\bf x}[i..i+n-1]$
is a palindrome and ${\bf x}[i..i+n-1]$ does not appear previously in 
$\bf x$

\smallskip

is equal to

\smallskip

\noindent the number of indices $i$ such that
${\bf x}[i..i+n-1] = {\bf x}[i..i+n-1]^R$ and for all $j$ with
$0 \leq j < i$,
${\bf x}[i..i+n-1]$ is not the same as ${\bf x}[j..j+n-1]$

\smallskip

is equal to

\smallskip

\noindent the number of indices $i$ such that
for all $t$, $0 \leq t \leq n/2$, $a(i+t)=a(i+n-1-t)$
and for all $j$ with $0 \leq j < i$, there exists $u$ with
$0 \leq u < n$
such that $a(i+u) \not= a(j+u)$.  Now apply Theorems~\ref{one} and \ref{equiv}.
\end{proof}

\begin{remark}
\normalfont
Allouche, Baake, Cassaigne, and Damanik
\cite{Allouche&Baake&Cassaigne&Damanik:2003}, Thm.\ 10,
proved that the palindrome 
complexity of the fixed point of a primitive $k$-uniform morphism
is $k$-automatic.  Our result is more general:  it shows that the
palindrome complexity of a $k$-automatic sequence is $k$-regular, and
hence is $k$-automatic iff it is bounded.

Jean-Paul Allouche kindly informs us that
our result has just been obtained independently by Carpi and D'Alonzo 
\cite{Carpi&DAlonzo:2011}.
\end{remark}

\begin{example}
\normalfont
Let $f(n)$ denote the number of unbordered factors of length
$n$ of the Thue-Morse sequence.
Here is a brief table of the values of $f(n)$:
\begin{table}[H]
\begin{center}
\begin{tabular}{|c|c|c|c|c|c|c|c|c|c|c|c|c|c|c|c|c|c|c|c|c|c|c|c|c|c|c|c|c|c|c|c|c|c|}
\hline
$n$   & 1 & 2 & 3 & 4 & 5 & 6 & 7 & 8 & 9 & 10 & 11 & 12 & 13 & 14 & 15 & 16 \\
\hline
$f(n)$& 2 & 2 & 4 & 2 & 4 & 6 & 0 & 4 & 4 & 4 & 4 & 12 & 0 & 4 & 4 & 8  \\
\hline
\end{tabular}
\end{center}
\end{table}
By Theorems~\ref{one} and \ref{equiv} we know that $f$ is $(\Enn,2)$-regular.  Conjecturally,
$f$ is given by the system of recurrences
\begin{eqnarray*}
f(4n+1) & = & f(2n+1) \\
f(8n+2) & =  &f(2n+1)-8f(4n) + f(4n+3) + 4f(8n) \\
f(8n+3) & =  &2f(2n) - f(2n+1) + 5f(4n) + f(4n+2) - 3f(8n) \\
f(8n+4) & =  &-4f(4n) + 2f(4n+2) + 2f(8n) \\
f(8n+6) & =  &2f(2n)-f(2n+1) + f(4n) + f(4n+2) + f(4n+3) -f(8n) \\
f(16n) & =  &-2f(4n) + 3f(8n) \\
f(16n+7) & = & -2f(2n) + f(2n+1) -5f(4n) + f(4n+2) +3f(8n) \\
f(16n+8) & = & -8f(4n) + 4f(4n+2) + 4f(8n) \\
f(16n+15) & = & -8f(4n) + 2f(4n+3) + 4f(8n) + f(8n+7) .
\end{eqnarray*}
In principle this could be verified by our method, but we have not
yet done so.
\end{example}

\begin{theorem}
Let ${\bf x} = a(0) a(1) a(2) \cdots$ be a $k$-automatic sequence.
Then the following sequences are also $k$-automatic:

\begin{itemize}
\item[(a)] $b(i) = 1$ if there is a square beginning at position $i$;
$0$ otherwise
\item[(b)] $c(i) = 1$ if there is a square centered at position $i$;
$0$ otherwise
\item[(c)] $d(i) = 1$ if there is an overlap beginning at position $i$;
$0$ otherwise
\item[(d)] $e(i) = 1$ if there is a palindrome beginning at position $i$;
$0$ otherwise
\item[(e)] $f(i) = 1$ if there is a palindrome centered at position $i$;
$0$ otherwise
\end{itemize}
\end{theorem}

\begin{remark}
\normalfont
Brown, Rampersad, Shallit, and Vasiga proved results (a)--(c) for the special
case of the Thue-Morse sequence \cite{Brown&Rampersad&Shallit&Vasiga:2006}.
\end{remark}

\begin{theorem}
Let $\bf x$ and $\bf y$ be $k$-automatic sequences.  Then the following are
$(\Enni, k)$-regular:

\begin{itemize}
\item[(a)] the number of distinct square factors in $\bf x$ of length $n$;
\item[(b)] the number of squares in $\bf x$ beginning at (centered at, ending at) position $n$;
\item[(c)] the length of the longest square in $\bf x$ beginning at (centered at,
ending at) position $n$;
\item[(d)] the number of palindromes in $\bf x$ beginning at (centered at, ending at) position $n$;
\item[(e)] the length of the longest palindrome in $\bf x$ beginning at (centered at,
ending at) position $n$;
\item[(f)] the length of the longest fractional power in $\bf x$  beginning at
(ending at) position $n$;
\item[(g)] the number of distinct recurrent factors in $\bf x$ of length $n$;
\item[(h)] the number of factors of length $n$ that occur in $\bf x$ but
not in $\bf y$.
\item[(i)] the number of factors of length $n$ that occur in both $\bf x$ and
$\bf y$.
\end{itemize}
\label{nine}
\end{theorem}

\begin{remark}
\normalfont
Brown, Rampersad, Shallit, and Vasiga proved results (b)--(c) for the special
case of the Thue-Morse sequence \cite{Brown&Rampersad&Shallit&Vasiga:2006}.
\end{remark}

We now turn to some other measures that have received much attention.
The recurrence function 
$R_{\bf x} (n) = R(n)$ of an infinite word ${\bf x}$ is the smallest integer $t$ such that every factor of
length $t$ of $\bf x$ contains as a factor every factor of length $n$.  Said otherwise,
it is the size of the smallest ``window'' one can slide along $\bf x$ and always
contain all length-$n$ factors.

\begin{theorem}
If ${\bf x}$ is $k$-automatic, then $(R_{\bf x} (n))_{n \geq 0}$ is $(\Enni,k)$-regular.
\end{theorem}

\begin{proof}
We translate the predicate ``$R(n) > t$'', as follows:

\noindent $R(n) > t$
\smallskip

iff

\smallskip

\noindent there exists $i \geq 0$, $j \geq 0$ such that
${\bf x}[j..j+n-1]$ appears nowhere in ${\bf x}[i..i+t-1]$

\smallskip

iff

\smallskip

\noindent there exists $i \geq 0$, $j \geq 0$ such that for all
integers $l$ with $i \leq l < i + t - 1 - n$ we have ${\bf x}[l..l+n-1] \not=
{\bf x}[j..j+n-1]$

\smallskip

iff

\smallskip

\noindent there exists $i \geq 0$, $j \geq 0$, such that for all
integers $l$ with $i \leq l < i + t - 1 - n$ there exists $m$,
$0 \leq m < n$ such that ${\bf x}[l+m] \not= {\bf x}[j+m]$.

Now for any fixed $n$, the number of positive integers $t$ for which 
$R(n) > t$ is equal to $R(n)$.  Hence $(R(n))_{n \geq 0}$ is $(\Enni, k)$-regular.
\end{proof}

Another measure is called ``appearance'' 
\cite{Allouche&Shallit:2003b}, \S 10.10.
The appearance function
$A_{\bf x} (n) = A(n)$ is the smallest integer $t$ such that every factor 
of length $n$ appears in a prefix of length $t$ of $\bf x$.  The following result
can be proved in an analogous manner to the previous one.

\begin{theorem}
If ${\bf x}$ is $k$-automatic, then $(A_{\bf x} (n))_{n \geq 0}$ is $(\Enn,k)$-regular.
\end{theorem}

Next, we consider a measure due to Garel \cite{Garel:1997}.
The separator length $S_{\bf x} (n)$ is the length of the smallest factor
that begins at position $n$ of $\bf x$ and does not occur previously.

\begin{theorem}
If ${\bf x}$ is $k$-automatic, then $(S_{\bf x} (n))_{n \geq 0}$ is $(\Enn,k)$-regular.
\end{theorem}

\begin{proof}
The predicate ``$S_{\bf x} (n) > t$'' is the same as saying that for every $i \leq t$
the word of length $i$ beginning at position $n$ of $\bf x$ occurs previously in
$\bf x$, which is the same as saying
for all $i, 0 \leq i \leq t$,  there exists $j, 0 \leq j < n$ such that
${\bf x}[n..n+i-1] = {\bf x}[j..j+i-1]$.
Now look at the pairs $(n,t)$ satisfying this, with $n$ positive.
For each $n$ there are exactly 
$S_{\bf x}(n)$ different $t$'s that work.  
\end{proof}

\begin{remark}
\normalfont
Garel \cite{Garel:1997} proved this for the case of a fixed point of a
uniform circular morphism; our proof works for the more general case of
an arbitrary $k$-automatic sequence.
\end{remark}

Carpi and D'Alonzo have introduced a measure they called 
repetitivity index \cite{Carpi&Dalonzo:2009}.  This measure $I_{\bf x} (n)$ is
the minimum distance between two
consecutive occurrences of the same length-$n$ factor in $\bf x$.  But
``$I_{\bf x}(n) > t$'' is the same as saying for all $i, j \geq 0$ with
$i \not= j$, the equality ${\bf x}[i..i+n-1] = {\bf x}[j..j+n-1]$ implies that
$j - i > t$.  Hence we get 

\begin{theorem}
If $\bf x$ is $k$-automatic, then its repetitivity index is $(\Enn,k)$-regular.
\end{theorem}

For our final application,
Frid and Zamboni \cite{Frid&Zamboni:2010} introduced the notion of
``automatic permutation''.  This is a permutation of $\Enn$ based on
a $k$-automatic sequence $\bf x$, as follows:  we say $i < j$ if
the infinite word ${\bf x}[i..\infty]$ is lexicographically less than
the word ${\bf x}[j..\infty]$.    The \textit{permutation complexity} 
$p_{\bf x} (n)$
is the map that sends $n$ to the number of distinct finite permutations
of length $n$ induced by $\bf x$ \cite{Frid:2011}.

\begin{theorem}
The permutation complexity  of a $k$-automatic sequence
is $(\Enn,k)$-regular.
\end{theorem}

\begin{proof}
First, we need to see that for $k$-automatic sequences the predicate
``the shift of $\bf x$ beginning at position $i$ 
is lexicographically less than the shift beginning at position $j$" is
$k$-automatic.

To see this, given positions $i$ and $j$, we verify that there is
some index $t$ such that $a[i+l] = a[j+l]$ for all $l < t$, and also that
$a[i+t] < a[j+t]$.

Next, we need to see that given $i, j, n$ we need to see that the
predicate ``the length-$n$ permutation induced by the shifts starting at
position $i$ coincides with that starting at $j$'' is automatic.

To do this we verify that
for all indices $l$ with $i \leq l, m < i+n$, the relation in the previous
paragraph holds between $i+l$ and $i+m$ in the same way
as it holds for $j+l$ and $j+m$.

In the final step, we enumerate the number of indices $i$ for which the
permutation at position $i$ of length $n$ does not match the one occurring
at any previous index.  This is just the number of distinct
permutations of length $n$.   
\end{proof}

As a corollary, we recover the result of Widmer \cite{Widmer:2011} that
the permutation complexity of
the Thue-Morse word is $(\Enn,2)$-regular.  In principle his description
could be mechanically verified.

\section{Linear bounds}

Yet another application of our method allows us to obtain linear bounds on
many quantities associated with automatic sequences.  As a first example,
we recover an old result of Cobham \cite{Cobham:1972} on ``subword'' complexity.

\begin{theorem}
The number of distinct factors of length $n$ of an
automatic sequence is $O(n)$.
\end{theorem}

\begin{proof} 
Let $\bf x$ be a $k$-automatic sequence.
By Theorem~\ref{logic} we know
that the base-$k$ encoding $S'$ of
\begin{eqnarray*}
S & =& \lbrace (n,I) \ :  \ 
\text{for all $j<I$ the
factor of length $n$ starting at position $j$} \\
&& \quad \quad  \text{is different from the one
starting at position $I$} \rbrace
\end{eqnarray*}
is a regular language.

Suppose that the
factor complexity of $\bf x$ is not $O(n)$.  Then for every $L$ there exists some
pair $(n,I) \in S$ such that the length of the canonical encoding of $I$ is longer
than that of $n$ by at least $L$ digits.  So in $S'$ there is some word of
the form $(n)_k {\tt B}^{\geq L} \times (I)_k$, 
where $(u)_k$ denotes the canonical
encoding of $u$ in base $k$ and $\times$ is how we join separate
components to form a word.

Since the length of $(I)_k$ is very
much longer than that of $(n)_k$, we can apply the pumping lemma to this
word, where we only pump in the portion of $(I)_k$ that is longer than
$(n)_k$.  Hence when we pump, we only add ${\tt B}$'s to the first component,
and so its value remains unchanged.  In this way by pumping we obtain
infinitely many values $I'$ such that $(n,I') \in S$.  In other words,
there are infinitely many distinct factors of length $n$, which is
clearly absurd.  The contradiction proves the result.
\end{proof}

In a similar manner we can prove that all the quantities in Theorem~\ref{nine}
are either linearly bounded, or unbounded.

\section{Other numeration systems}

All our results transfer, {\it mutatis mutandis}, to the setting of other
numeration systems where addition can be performed on numbers using a
transducer that processes numbers starting with the least significant
digit.

A (generalized) numeration system is given by an increasing sequence of integers 
$U=(U_i)_{i\ge0}$ such that $U_0=1$ and $C_U:=\lim_{i\to +\infty} U_{i+1}/U_i$ exists and is finite. 
Then the canonical $U$-representation of $n$ (with least significant digit first), which is denoted by $(n)_U$,
is the unique finite word $w$ over the alphabet $\Sigma_U=\{0,\ldots,C_U-1\}$ not ending with $0$ and satisfying
$n=\sum_{i=0}^{|w|-1}w[i]\,U_i\ \text{ and }\ \forall t\in\{0,\ldots,|w|-1\},\ \sum_{i=0}^t w[i]\,U_i<U_{t+1}.$
The notion of $k$-automatic sequence extends naturally to this context: 
an infinite sequence {\bf x} is said to be $U$-automatic if it is computable by a finite automaton
taking as input the $U$-representation $(n)_U$ of $n$, and
having ${\bf x}[n]$ as the output associated with the last state encountered.

A numeration system $U$ is called {\it linear} if $U$ satisfies a linear recurrence relation over $\mathbb Z$.
A Pisot system is a linear numeration system $U$ whose characteristic polynomial is the 
minimal polynomial of a Pisot number. Recall that a Pisot number is an algebraic integer greater 
than 1, all of whose conjugates have moduli less than 1. 
For example, all integer base numeration systems and the Fibonacci numeration system are Pisot systems.
Frougny and Solomyak \cite{Frougny&Solomyak:1996}
proved that addition is $U$-recognizable within all Pisot systems $U$, i.e., it can be performed by a finite letter-to-letter transducer reading $U$-representations with least significant digit first.
Bruy\`ere and Hansel \cite{Bruyere&Hansel:1997} then proved the following 
logical characterization of $U$-automatic sequences for Pisot systems: 
a sequence is $U$-automatic if and only if it is $U$-definable,
i.e., it is expressible as a predicate of $\langle \Enn, +, V_U \rangle$, 
where $V_U(n)$ is the smallest $U_i$ occurring in $(n)_U$ with a nonzero coefficient. 
Therefore, if $U$ is a Pisot system, any combinatorial property of $U$-automatic words that can be 
described by a predicate of $\langle \Enn, +, V_U \rangle$ is decidable. 

The notion of $(R,k)$-regular sequences extends to Pisot numeration
systems:  an infinite sequence ${\bf x}$ is said to be {\it $(R,U)$-regular} if
the series $\sum_{n \geq 0}  {\bf x}[n] (n)_U$ is an
$R$-recognizable series. Thus we obtain 

\begin{theorem} \label{the:Fibonacci}
Let $U$ be a Pisot numeration system and let ${\bf x}$ be any $U$-automatic word.
The following sequences are $U$-automatic:

\begin{itemize}
\item[(a)] $a(n) = 1$ if there is a square beginning at (centered at, ending at)
position $n$ of $\bf x$, $0$ otherwise;
\item[(b)] $b(n) = 1$ if there is a palindrome beginning at (centered at,
ending at) position $n$ of $\bf x$, $0$ otherwise;
\item[(c)] $c(n) = 1$ if there is an unbordered factor beginning at
(centered at, ending at) position $n$ of $\bf x$, $0$ otherwise.
\end{itemize}

The following sequences are $(\Enni,U)$-regular: 

\begin{itemize}
\item[(a)] The number of distinct square factors
beginning at (centered at, ending at)
position $n$ of $\bf x$;
\item[(b)] The number of distinct palindromic
factors beginning at (centered at,
ending at) position $n$ of $\bf x$, $0$ otherwise;
\item[(c)] The number of distinct unbordered factors
beginning at
(centered at, ending at) position $n$ of $\bf x$, $0$ otherwise.
\end{itemize}

\end{theorem}

Berstel showed that the cardinality of the set of unnormalized
Fibonacci representations is Fibonacci-regular \cite{Berstel:2001},
a result also obtained (but not published) by the third author about
the same time.  In analogy with Theorem~\ref{ns} we have

\begin{theorem}
The number of unnormalized representations of $n$ in a Pisot
numeration system $U$ is $(\Enni,U)$-regular.
\end{theorem}

\section{Closing remarks}

It may be worth noting that the explicit constructions of automata we have
given also imply bounds on the smallest example of (or counterexample to) the
properties we consider.  The bounds are essentially given by a tower
of exponents whose height is related to the number of alternating
quantifiers.  For example,

\begin{theorem}
Suppose $\bf x$ and $\bf y$ are $k$-automatic sequences generated by
automata with at most $q$ states.  If the set of factors of $\bf x$
differs from the set of factors of $\bf y$, then there exists a factor
of length at most $2^{2^{2^{2q^2}}}$ that occurs in one word but not
the other.
\end{theorem}

We also note that
a question left open in \cite{Allouche&Rampersad&Shallit:2009}, regarding
the description of the lexicographically least word in the orbit closure
of the Rudin-Shapiro sequence, was recently solved by
Currie \cite{Currie:2010}.

Finally, 
in a recent paper \cite{Shallit:2011}, the third author shows that additional
properties of automatic sequences are deducible by expanding on the techniques
in this paper.  For example, the critical exponent is computable.

\section{Acknowledgments}

We thank Jean-Paul Allouche for his helpful comments.


\begin{thebibliography}{10}

\bibitem{Allouche&Baake&Cassaigne&Damanik:2003}
J.-P. Allouche, M.~Baake, J.~Cassaigne, and D.~Damanik.
\newblock Palindrome complexity.
\newblock {\em Theoret. Comput. Sci.} {\bf 292} (2003), 9--31.

\bibitem{Allouche&Rampersad&Shallit:2009}
J.-P. Allouche, N.~Rampersad, and J.~Shallit.
\newblock Periodicity, repetitions, and orbits of an automatic sequence.
\newblock {\em Theoret. Comput. Sci.} {\bf 410} (2009), 2795--2803.

\bibitem{Allouche&Shallit:1992}
J.-P. Allouche and J.~O. Shallit.
\newblock The ring of $k$-regular sequences.
\newblock {\em Theoret. Comput. Sci.} {\bf 98} (1992), 163--197.

\bibitem{Allouche&Shallit:2003a}
J.-P. Allouche and J.~O. Shallit.
\newblock The ring of $k$-regular sequences, {II}.
\newblock {\em Theoret. Comput. Sci.} {\bf 307} (2003), 3--29.

\bibitem{Allouche&Shallit:2003b}
J.-P. Allouche and J.~Shallit.
\newblock {\em Automatic Sequences: Theory, Applications, Generalizations}.
\newblock Cambridge University Press, 2003.

\bibitem{Berstel:2001}
J.~Berstel.
\newblock An exercise on {Fibonacci} representations.
\newblock {\em RAIRO Inform. Th\'eor. App.} {\bf 35} (2001), 491--498.

\bibitem{Berstel&Reutenauer:2011}
J.~Berstel and C.~Reutenauer.
\newblock {\em Noncommutative Rational Series With Applications}, Vol. 137 of
  {\em Encyclopedia of Mathematics and Its Applications}.
\newblock Cambridge University Press, 2011.

\bibitem{Brown&Rampersad&Shallit&Vasiga:2006}
S.~Brown, N.~Rampersad, J.~Shallit, and T.~Vasiga.
\newblock Squares and overlaps in the {Thue-Morse} sequence and some variants.
\newblock {\em RAIRO Inform. Th\'eor. App.} {\bf 40} (2006), 473--484.

\bibitem{Bruyere&Hansel:1997}
V.~{Bruy\`ere} and G.~Hansel.
\newblock Bertrand numeration systems and recognizability.
\newblock {\em Theoret. Comput. Sci.} {\bf 181} (1997), 17--43.

\bibitem{Bruyere&Hansel&Michaux&Villemaire:1994}
V.~{Bruy\`ere}, G.~Hansel, C.~Michaux, and R.~Villemaire.
\newblock Logic and $p$-recognizable sets of integers.
\newblock {\em Bull. Belgian Math. Soc.} {\bf 1} (1994), 191--238.
\newblock Corrigendum, {\it Bull.\ Belg.\ Math.\ Soc.} {\bf 1} (1994), 577.

\bibitem{Carpi&Dalonzo:2009}
A.~Carpi and V.~D'Alonzo.
\newblock On the repetitivity index of infinite words.
\newblock {\em Internat. J. Algebra Comput.} {\bf 19} (2009), 145--158.

\bibitem{Carpi&DAlonzo:2011}
A.~Carpi and V.~D'Alonzo.
\newblock On factors of synchronized sequences.
\newblock To appear, {\it Theor. Comput. Sci.}, 2011.

\bibitem{Carpi&Maggi:2001}
A.~Carpi and C.~Maggi.
\newblock On synchronized sequences and their separators.
\newblock {\em RAIRO Inform. Th\'eor. App.} {\bf 35} (2001), 513--524.

\bibitem{Cobham:1972}
A.~Cobham.
\newblock Uniform tag sequences.
\newblock {\em Math. Systems Theory} {\bf 6} (1972), 164--192.

\bibitem{Currie:2010}
J.~D. Currie.
\newblock Lexicographically least words in the orbit closure of the
  {Rudin-Shapiro} word.
\newblock {\tt http://arxiv.org/pdf/0905.4923}, 2010.

\bibitem{Currie&Saari:2009}
J.~D. Currie and K.~Saari.
\newblock Least periods of factors of infinite words.
\newblock {\em RAIRO Inform. Th\'eor. App.} {\bf 43} (2009), 165--178.

\bibitem{Fagnot:1997a}
I.~Fagnot.
\newblock Sur les facteurs des mots automatiques.
\newblock {\em Theoret. Comput. Sci.} {\bf 172} (1997),
  67--89.

\bibitem{Frid:2011}
A.~Frid.
\newblock Infinite permutations vs. infinite words.
\newblock In P.~Ambro{\u{z}}, S.~Holub, and Z.~{Mas\'akov\'a}, editors, {\em
  WORDS 2011, 8th International Conference}. Elect. Proc. Theor. Comput. Sci.,
  2011.
\newblock Available at {\tt http://arxiv.org/abs/1108.3616v1}.

\bibitem{Frid&Zamboni:2010}
A.~Frid and L.~Q. Zamboni.
\newblock On automatic infinite permutations.
\newblock Presented at {\it Journ\'ees Montoises}, 2010.

\bibitem{Frougny&Solomyak:1996}
C.~Frougny and B.~Solomyak.
\newblock On representation of integers in linear numeration systems.
\newblock In M.~Pollicott and K.~Schmidt, editors, {\em Ergodic Theory of
  {$\Zee^d$} Actions (Warwick, 1993--1994)}, Vol. 228 of {\em London
  Mathematical Society Lecture Note Series}, pp.  345--368. Cambridge
  University Press, 1996.

\bibitem{Garel:1997}
E.~Garel.
\newblock {S\'eparateurs} dans les mots infinis {engendr\'es} par morphismes.
\newblock {\em Theoret. Comput. Sci.} {\bf 180} (1997), 81--113.

\bibitem{Halava&Harju&Karki&Rigo:2010}
V.~Halava, T.~Harju, T.~{K\"arki}, and M.~Rigo.
\newblock On the periodicity of morphic words.
\newblock In {\em Developments in Language Theory 2010}, Vol. 6224 of {\em
  Lecture Notes in Computer Science}, pp.  209--217. Springer-Verlag, 2010.

\bibitem{Honkala:1986}
J.~Honkala.
\newblock A decision method for the recognizability of sets defined by number
  systems.
\newblock {\em RAIRO Inform. Th\'eor. App.} {\bf 20} (1986), 395--403.

\bibitem{Krieger&Shallit:2007}
D.~Krieger and J.~Shallit.
\newblock Every real number greater than $1$ is a critical exponent.
\newblock {\em Theoret. Comput. Sci.} {\bf 381} (2007), 177--182.

\bibitem{Kuich&Salomaa:1986}
W.~Kuich and A.~Salomaa.
\newblock {\em Semirings, Automata, Languages}.
\newblock Springer-Verlag, 1986.

\bibitem{Leroux:2005}
J.~Leroux.
\newblock A polynomial time {Presburger} criterion and synthesis for number
  decision diagrams.
\newblock In {\em 20th IEEE Symposium on Logic in Computer Science (LICS
  2005)}, pp.  147--156. IEEE Press, 2005.

\bibitem{Mosse:1996a}
B.~{Moss\'e}.
\newblock {Reconnaissabilit\'e} des substitutions et {complexit\'e} des suites
  automatiques.
\newblock {\em Bull. Soc. Math. France} {\bf 124} (1996),
  329--346.

\bibitem{Nicolas&Pritykin:2009}
F.~Nicolas and Yu. Pritykin.
\newblock On uniformly recurrent morphic sequences.
\newblock {\em Internat. J. Found. Comp. Sci.} {\bf 20} (2009), 919--940.

\bibitem{Reznick:1990}
B.~Reznick.
\newblock Some binary partition functions.
\newblock In {\em Analytic Number Theory}, Vol.~85 of {\em Progr. Math.}, pp.
  451--477. {Birkh\"auser}, 1990.

\bibitem{Saari:2008}
K.~Saari.
\newblock {\em On the Frequency and Periodicity of Infinite Words}.
\newblock PhD thesis, University of Turku, Finland, 2008.

\bibitem{Sakarovitch:2009}
J.~Sakarovitch.
\newblock {\em Elements of Automata Theory}.
\newblock Cambridge University Press, 2009.

\bibitem{Salomaa&Soittola:1978}
A.~Salomaa and M.~Soittola.
\newblock {\em Automata-Theoretic Aspects of Formal Power Series}.
\newblock Springer-Verlag, 1978.

\bibitem{Schutzenberger:1962}
M.-P. {Sch\"utzenberger}.
\newblock On a theorem of {R. Jungen}.
\newblock {\em Proc. Amer. Math. Soc.} {\bf 13} (1962), 885--890.

\bibitem{Shallit:2011}
J.~Shallit.
\newblock The critical exponent is computable for automatic sequences.
\newblock In P.~Ambro{\u{z}}, S.~Holub, and Z.~{Mas\'akov\'a}, editors, {\em
  WORDS 2011, 8th International Conference}. Elect. Proc. Theor. Comput. Sci.,
  2011.
\newblock Available at {\tt http://arxiv.org/abs/1104.2303v2}.

\bibitem{Widmer:2011}
S.~Widmer.
\newblock Permutation complexity of the {Thue-Morse} word.
\newblock {\em Adv. in Appl. Math.} {\bf 47} (2011), 309--329.

\end{thebibliography}
\end{document}